\documentclass[11pt, a4paper]{article}
\pdfoutput=1 
\usepackage[english]{babel}
\usepackage{appendix}
\usepackage{graphicx}
\usepackage[margin=1.2in]{geometry}
\graphicspath{{./img}}
\usepackage{amssymb}
\usepackage{amsthm}
\usepackage{amsmath}
\usepackage{thmtools, thm-restate}
\usepackage{mathtools}
\usepackage{mathrsfs}
\usepackage[notrig]{physics}
\usepackage{xcolor}
\usepackage{tikz}
\usetikzlibrary{quantikz}
\usepackage{fancybox}
\usepackage{floatrow}
\usepackage{mathtools}
\usepackage[format=plain,labelfont={bf,it},textfont=it]{caption}
\usepackage{subcaption}
\usepackage{algorithmicx}
\usepackage{algorithm,algpseudocode}
\usepackage{dsfont}
\usepackage{diagbox,multirow}
\usepackage{enumerate}
\usepackage{academicons}

\usepackage[%
backend=bibtex,
style=numeric-comp,
firstinits=true,
sorting=none,
]{biblatex}
\usepackage{hyperref} 
\usepackage[nameinlink,noabbrev]{cleveref}

\renewcommand*{\arraystretch}{1.2}

\DeclareMathAlphabet\euscr{U}{eus}{m}{n}

\newcommand{\sboqc}{\euscr{S}}
\newcommand{\sboqcp}{\euscr{S'}}

\newcommand{\vera}{V_{\mathcal A}}
\newcommand{\vero}{V_{\mathcal O}}
\newcommand{\res}{\euscr{R}}

\newcommand{\neig}[1]{N_{\mathcal G}(#1)}
\newcommand{\neigc}[1]{N_{\mathcal G}[#1]}
\newcommand{\invf}[1]{\mathop{}\!{invf}(#1)}
\newcommand{\conv}[1]{\pi_{\mathcal{#1}}}
\newcommand{\simb}{\sigma_{\mathcal B}}
\newcommand*{\scaleeq}[2][4]{\scalebox{#1}{\ensuremath{#2}}}%

\usetikzlibrary{decorations.markings,positioning,arrows,decorations.pathmorphing,shapes,patterns}
\usetikzlibrary{decorations.markings}
\tikzstyle{abo}=[rectangle,rounded corners,minimum width=0,minimum height=0,text centered, draw=black]
\tikzstyle{bb}=[dashed,rectangle,text centered,minimum width=2.7cm,minimum height=1.95cm,draw=black]
\tikzstyle{bbk}=[dashed,rectangle,text centered,minimum width=2.15cm,minimum height=1.65cm,draw=black]
\tikzstyle{bbkd}=[dotted,rectangle,text centered,minimum width=2.15cm,minimum height=1.65cm,draw=black]

\tikzset{split fill/.style args={#1 and #2}{path picture={
    \fill [#1] (path picture bounding box.south west)
      rectangle (path picture bounding box.north);
    \fill [#2] (path picture bounding box.south)
      rectangle (path picture bounding box.north east);
}}}

\colorlet{gay}{black!70!white}

\theoremstyle{plain}
\floatstyle{ruled}

\floatstyle{ruled}
\newfloat{protocol}{htbp}{lok}
\floatname{protocol}{Protocol}
\newfloat{subroutine}{htbp}{lok}
\floatname{subroutine}{Subroutine}
\crefname{algorithm}{Algorithm}{Algorithms}
\crefname{subroutine}{Subroutine}{Subroutines}
\crefname{protocol}{Protocol}{Protocols}
\crefname{definition}{Definition}{Definitions}
\crefname{theorem}{Theorem}{Theorems}
\crefname{corollary}{Corollary}{Corollary}
\newtheorem{conjecture}{Conjecture}
\crefname{conjecture}{Conjecture}{Conjecture}
\newtheorem{definition}{Definition}
\newtheorem{theorem}{Theorem}
\newtheorem{proposition}{Proposition}

\crefname{lemma}{Lemma}{Lemmas}

\newtheorem{example}{Example}
\crefname{example}{Example}{Examples}

\algnewcommand\algorithmicinput{\textbf{Input:}}
\algnewcommand\Input{\item[\algorithmicinput]}
\algnewcommand\algorithmicoutput{\textbf{Output:}}
\algnewcommand\Output{\item[\algorithmicoutput]}

\DeclareRobustCommand{\squared}[1]{\tikz[baseline=(char.base)]{
        \node[shape=rectangle,draw=gay,inner sep=0.5pt,fill=gay,text=white] (char)
{\normalfont #1};}}

\DeclareRobustCommand{\circled}[1]{\tikz[baseline=(char.base)]{
        \node[shape=circle,draw=gay,inner sep=0.5pt,fill=gay,text=white] (char)
{\normalfont #1};}}

\DeclareRobustCommand{\circledl}[1]{\tikz[baseline=(char.base)]{
        \node[shape=circle,draw=gay,inner sep=0pt,fill=gay,text=white, minimum size=10pt] (char)
{\ttfamily #1};}}

\addto\extrasenglish{    }

\makeatletter
\algrenewcommand\ALG@beginalgorithmic{\footnotesize}
\makeatother

\usepackage{tabularx,ragged2e,booktabs}
\newcolumntype{L}[1]{>{\raggedright\let\newline\\\arraybackslash\hspace{0pt}}m{#1}}
\newcolumntype{C}[1]{>{\centering\let\newline\\\arraybackslash\hspace{0pt}}m{#1}}
\newcolumntype{R}[1]{>{\raggedleft\let\newline\\\arraybackslash\hspace{0pt}}m{#1}}

\newcommand{\etal}{\textit{et al.}~}
\newcommand{\ie}{\textit{i}.\textit{e}.,~}
\newcommand{\eg}{\textit{e}.\textit{g}.,~}
\newcommand{\viz}{\textit{viz.}~}

\tikzset{
  vertex/.style={fill,draw,inner sep=0pt,label distance=1pt,minimum
  size=4pt,circle},
   clusterm/.style={
     column sep=.7cm, row sep=.7cm,
     matrix of nodes,
     nodes=vertex}
}
\newsavebox{\circbox}
\usetikzlibrary{fit,calc,positioning,decorations.pathreplacing,matrix}
\newcommand*{\tightdisplaymath}{\abovedisplayskip\z@\belowdisplayskip\z@}

\newcommand{\boxsize}{.3cm}
\newcommand{\boxabc}[1]%
  {\ovalbox{\text{\begin{minipage}{\boxsize}\centering #1\end{minipage}}}}

\DeclareMathOperator*{\OP}{%
\mathchoice%
{\ooalign{\raisebox{.05\height}{$\,\displaystyle\prod$}\cr 
{$\displaystyle\longrightarrow$}\cr\hidewidth{$\displaystyle\longrightarrow$}\cr\hidewidth
\raisebox{\height}{\scalebox{.7}{$\textstyle >$\!\!\!}}
{$\displaystyle\phantom{\prod}$}}}
{\ooalign{{$\,\textstyle\prod$}\cr 
    \scalebox{.9}{$\textstyle\longrightarrow$}\cr\hidewidth\scalebox{.9}{$\textstyle\longrightarrow$}\cr\hidewidth
\raisebox{\height}{\scalebox{.6}{\!\!\!$\textstyle >$}}\hidewidth
{$\textstyle\phantom{\prod}$}}}
  {\ooalign{\raisebox{.1\height}{\scalebox{.75}{$\scriptstyle\prod$}}\cr$\scriptstyle\rightarrow$\cr}}
  {\ooalign{\raisebox{.1\height}{\scalebox{.75}{$\scriptstyle\prod$}}\cr$\scriptstyle\rightarrow$\cr}}
}

\DeclareMathOperator*{\POP}{%
\mathchoice%
{\ooalign{\raisebox{.05\height}{$\,\displaystyle\prod$}\cr 
{$\displaystyle\longrightarrow$}\cr\hidewidth{$\displaystyle\longrightarrow$}\cr\hidewidth
\raisebox{\height}{\scalebox{.7}{$\textstyle \succ$\!\!\!}}
{$\displaystyle\phantom{\prod}$}}}
{\ooalign{{$\,\textstyle\prod$}\cr 
    \scalebox{.9}{$\textstyle\longrightarrow$}\cr\hidewidth\scalebox{.9}{$\textstyle\longrightarrow$}\cr\hidewidth
\raisebox{\height}{\scalebox{.6}{\!\!\!$\textstyle \succ$}}\hidewidth
{$\textstyle\phantom{\prod}$}}}
  {\ooalign{\raisebox{.1\height}{\scalebox{.75}{$\scriptstyle\prod$}}\cr$\scriptstyle\rightarrow$\cr}}
  {\ooalign{\raisebox{.1\height}{\scalebox{.75}{$\scriptstyle\prod$}}\cr$\scriptstyle\rightarrow$\cr}}
}

\hypersetup{
    colorlinks,
    linkcolor={blue!60!black},
    citecolor={green!50!black},
    urlcolor={red!50!black}
}

\newcommand{\orcid}[1]{\href{https://orcid.org/#1}{\includegraphics[width=8pt]{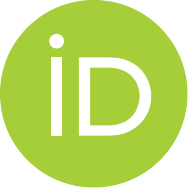}}}

\author{\orcid{0000-0003-0558-4685}\,Cica Gustiani\thanks{\texttt{cica.gustiani@materials.ox.ac.uk}}~${}^{1,3}$
    and \orcid{0000-0003-4332-645X}\,David P. DiVincenzo${}^{1,2}$
}  
\date{%
    {\small\em ${}^{1}$Institute for Quantum Information, RWTH Aachen University, D-52056 Aachen, Germany\\
        ${}^{2}$Peter Grünberg Institute, Theoretical Nanoelectronics, Forschungszentrum Jülich, D-52425 Jülich, Germany\\
     Jülich-Aachen Research Alliance (JARA), Fundamentals of Future Information Technologies, D-52425 Jülich, Germany\\
${}^{3}$Department of Materials, University of Oxford, Parks Road, Oxford OX1 3PH, United Kingdom}\\[2ex]
    \today
}

\title{Blind Oracular Quantum Computation}

\begin{document}
\maketitle
\begin{abstract}
In the standard oracle model, an oracle efficiently evaluates an unknown
classical function independent of the quantum algorithm itself. 
Quantum
algorithms have a complex interrelationship to their oracles; for
example the possibility of quantum speedup is affected by the manner by which oracles are implemented.
Therefore,
it is physically meaningful to separate oracles from their quantum algorithms, and we introduce one such separation here.
We define the {\em Blind Oracular Quantum Computation} (BOQC) scheme, in which the oracle is a distinct node in a quantum network.
Our work augments the client-server setting of quantum computing, in which a powerful quantum computer server is available on the network for discreet use by clients on the network with low quantum power. In BOQC, an oracle is another client that cooperates with the main client so that an oracular quantum algorithm is run on the server. The cooperation between 
the main client and the oracle takes place (almost)
without communication.  
We prove BOQC to be {\em blind}:
the server cannot learn anything about the
clients' computation. This proof is performed within the composable security definitions provided by the formalism of Abstract Cryptography.
We enhance the BOQC scheme to be runnable with
minimal physical qubits when run on a solid-state quantum network; we prove that this
scheme, which we refer to as BOQCo (BOQC-optimized), possesses the same security as
BOQC. 


\end{abstract}

\section{Introduction}
Oracle constructions in quantum algorithms provide an essential conceptual
framework for understanding quantum speedups.  The detailed interrelationship
between oracle properties and algorithmic efficiency is complex: an interesting
example arises in the Grover algorithm, where the quantum speedup becomes
impossible if the oracle has a small probability of failing on every
call~\cite{regev2008impossibility}.  Moreover, for some quantum algorithms,
adding internal dice to an oracle introduces a strong separation. For example, Simon's
algorithm, when the oracle has internal dice, is unsolvable on classical computers
while it is solvable in a linear time on quantum
    computers~\cite{harrow2011uselessness}.  In the standard model, a quantum
    oracle is specified as a unitary map $\ket{x,y}\mapsto\ket{x,y\oplus
    f(x)}$, where $f:\{0,1\}^n\rightarrow\{0,1\}$ indicates a subroutine whose
    code we cannot usefully examine or a ``black box'' whose properties we
    would like to estimate.\footnote{An equivalent oracle model, reducible to
   the standard oracle model, is the \emph{phase oracle}, represented as a map
    $\ket{x,y}\mapsto (-1)^{y.f(x)}\ket{x,y}$~\cite{zhandry2018how}.  In the
    standard oracle model, $f$ is a deterministic classical function; however,
    some generalizations are introduced
    in~\cite{meyer2009single,harrow2011uselessness}, where $f$ can be a
    probabilistic classical function.} While oracle constructions have been
    considered artificial, we aim to introduce and analyze a multiparty setting
    for which the oracle paradigm is physically meaningful. 

\begin{figure}[!t]
\includegraphics[scale=0.15]{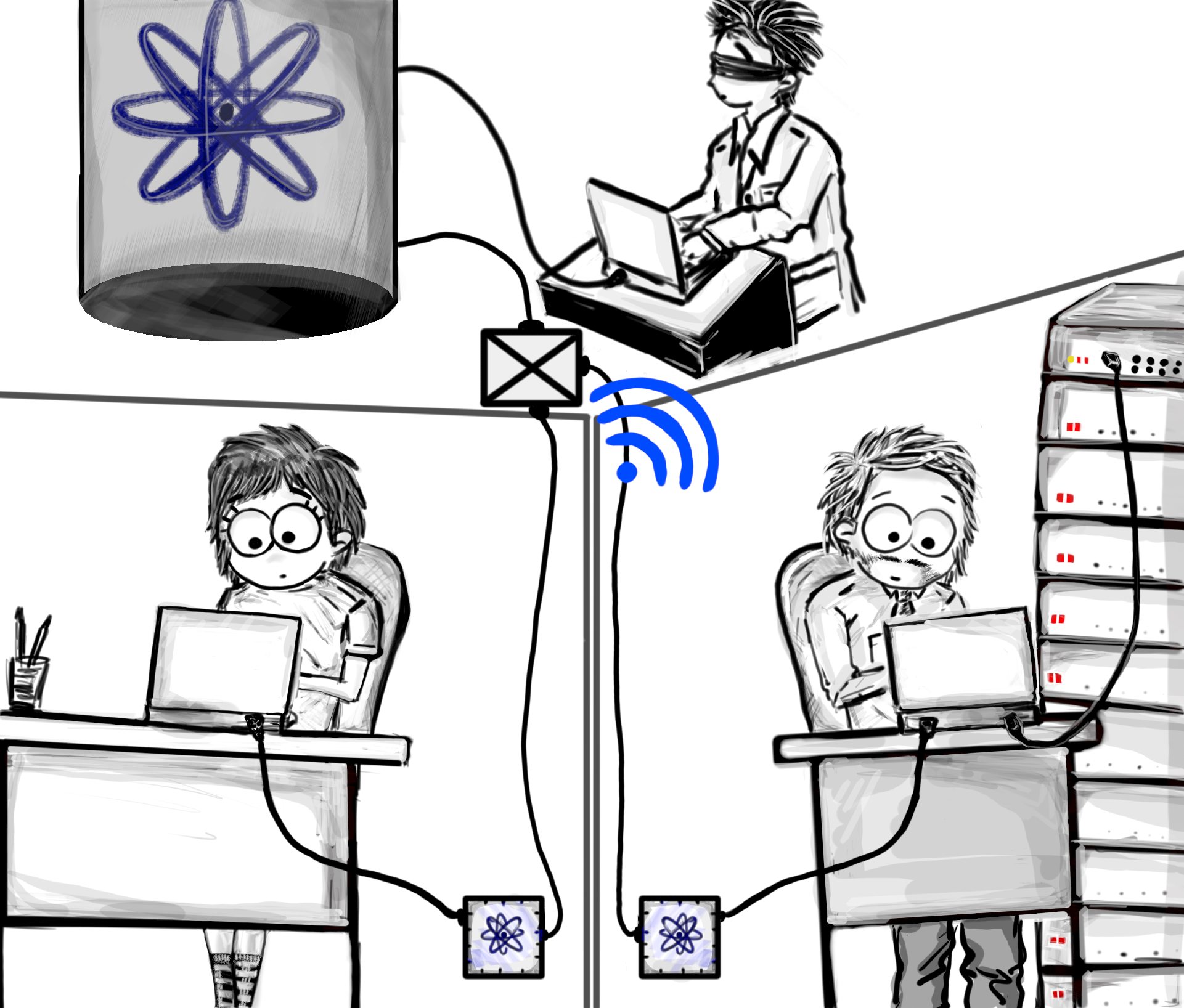}
\caption{\label{fig:cartoon}\cite{gustiani}
Alice (left) wants to run a private oracular quantum algorithm, but she
has neither a powerful quantum computer nor the knowledge needed to construct her
oracle function. Oscar (right) is equipped with a device that evaluates the oracle
functions, \eg it computes stocks' prediction, gives access to a database, or performs
massive classical computations. Bob (top), who owns a powerful quantum computer,
offers a quantum computing service. However, Alice and Oscar do not trust Bob's
quantum computer, and 
they do not risk revealing information about their computation while running it. 
Therefore, they execute the computation within the BOQC scheme in which Bob is
blinded to the computation.
All players are connected to an insecure classical channel (the wireless). 
Alice and Oscar still need a small quantum power --- the small boxes in the bottom, making their computation power
``almost classical.'' Such power is necessary because blindness
is impossible for a purely classical client~\cite{morimae2014impossibility}.
}
\end{figure}

Our view of the near-term situation in the development of quantum processing is
that there will be quantum computers of some moderate power (the {\em servers}
of the discussion below) owned by particular organizations which offer their
service on a quantum internet to clients with small quantum power.  We will
further assume that there will be nodes on the quantum internet, also with
modest quantum processing power, but in possession of some special information
or data --- the {\em oracles}. Consider \Cref{fig:cartoon} for an illustration.
We introduce a protocol where a client (Alice),
who is aware of these oracle resources on the network, needs to delegate her
oracular quantum algorithm to an untrustworthy server.  We refer to our scheme as
Blind Oracular Quantum Computation (BOQC). The BOQC protocol views the oracle
as a trustworthy third party.\footnote{We also consider a party with oracle
access as a party with more classical power.} We will consider the server quantum computer (Bob) to be an
untrustworthy party, therefore we adopt the concept of
\emph{blindness}~\cite{childs2001secure,broadbent2009universal,fitzsimons2017private,morimae2012blind},
where the server can learn nothing about the algorithm that is running and nothing
about the measurement outcomes.  Our protocol is set to run any of the family
of quantum oracular algorithms, an extensive catalog of which can be found
in~\cite{zoo}.

BOQC is an extension of \emph{Universal Blind Quantum Computation}
(UBQC)~\cite{broadbent2009universal}, enabling computation controlled by two
cooperating clients (the main client and the oracle) that are not interacting
during the execution of the computation. We will show that BOQC is
\emph{composably blind} using the \emph{Abstract Cryptography} (AC) framework
\cite{maurer2011abstract}; references closely related to our work are~\cite{dunjko2014composable,christopher}. UBQC allows the desired delegation
of computation by a client to an untrustworthy server, where the algorithms are
implemented within a measurement-based computation model, that is the
\emph{One-way Quantum Computer}
(1WQC)~\cite{raussendorf2000quantum,raussendorf2003measurement}.  The essential
resource of 1WQC is a highly entangled qubit state, \ie a graph state or a
cluster state.

The idea of our choice to use the 1WQC model for our protocol stems from an
important property of quantum maps over a graph state: combining two quantum
operators (graphs) is performed by simply connecting the graphs. Thus,
outsourcing a quantum operation to another party means outsourcing the
corresponding graph. This outsourcing is not so straightforward in the
conventional gate-based quantum computation model.

Initially, the 1WQC paradigm was
tailored to perform quantum computations on ultracold atoms in optical
lattice systems on which the resource state (cluster states) can be generated
efficiently~\cite{jaksch1999entanglement}. In the 1WQC paradigm, the computation
is done by systematically ``consuming'' the resource via measurements --- thus, it
is called a ``one-way computer.'' It turns out that the 1WQC paradigm can also
be efficiently performed on linear optics systems, whose qubits are memoryless.
Several small-scale 1WQC computations have been experimentally conducted on
linear optical
systems~\cite{walther2005experimental,prevedel2007high,chen2007experimental,tame2007experimental,barz2012demonstration,greganti2016demonstration}.


However, performing adaptive measurement using memoryless qubits remains
challenging. In our companion paper~\cite{gustiani2019three}, we have sought
possibilities beyond memoryless qubits by considering solid-state qubits,
since here performing adaptive measurements can be more feasible. However,
producing graph states is very demanding, \eg running an exact quantum search
algorithm within a BOQC scheme needs more than 90 qubits for searching one item
within five~\cite{gustiani2019three}.  Those motivate us to push forward the
implementations of 1WQC using solid-state qubits and to diminish the massive
requirement for physical qubits: we propose to prepare a graph state by parts,
in a ``just-in-time'' fashion in which qubits are prepared only when
needed.\footnote{The just-in-time fashion has the same principle as the
\emph{lazy computation} scheme that is common programming practice.} We call our scheme
\emph{lazy 1WQC}, for which we require qubits that possess permanence and can be
rapidly re-initialized. 

Note that, the reusing-qubits scheme has been analyzed in
in~\cite{houshmand2018minimal}, which obtains bounds on the number of physical qubits needed.
However, here, we provide lazy 1WQC as a clear-cut and explicit scheme that is
provably correct. Then, we integrate lazy 1WQC into the BOQC scheme,
producing BOQCo (BOQC-optimized). BOQCo is BOQC in which the number of physical
qubits is minimal, given that Bob's quantum computer is a solid-state system
whose qubits possess permanence and can be rapidly re-initialized. In the BOQCo
scheme, the exact search of one item within five exactly requires 4 physical
qubits~\cite{gustiani2019three}. 

\section{Preliminaries}
\subsection{Measurement-based computation and formalism}
\label{sec:1wqc}

Here we review \emph{one-way quantum computing}
(1WQC)~\cite{raussendorf2000quantum,raussendorf2003measurement}, the
measurement calculus~\cite{danos2007measurement}, and deterministic
computations made possible by the existence of
\emph{flow}~\cite{danos2006determinism}.  A computation within 1WQC is executed
by consecutively measuring qubits in a \emph{cluster state}: a highly entangled
quantum state, which can be efficiently parameterized by mathematical
graphs~\cite{briegel2009measurement}.  A cluster state corresponds to the
space-time layout of the quantum computer, and consecutive measurements define
quantum operations. In this study, we use measurement calculus to describe
processes within a 1WQC computation, and we use flow to specify
measurement-dependency structure.

A graph is used to represent the resource (cluster state) of a 1WQC
computation.  The graph's vertices represent qubits whose states are initially
in the $xy$-plane of Bloch sphere, and its edges represent CPHASE gates applied
to the corresponding nodes.  We will interchangeably use graph and graph state.
In particular, \emph{open graph} is used as a 1WQC resource, that is a triplet
$(\mathcal G, I, O)$, with a set of quantum input nodes $I$ and quantum output
nodes $O$ that may intersect, where $\mathcal G=(V, E)$ is a simple
graph\footnote{Simple graphs are a class of graphs without direction, without
self loops, and without multiple edges.} with a set of vertices $V$ and edges
$E$, $I\subset V$ and $O\subset V$, where $I\neq\varnothing$, and $O\neq\varnothing$. 
For a subset $K\subset V$, $\mathcal G[K]$
denotes the induced subgraph whose vertex set is $K$ and whose edge set is
those from $E$ whose endpoints are both in $K$ --- denoted as $E(K)$. 
Given a node $k$, we define the following
notations:
\begin{equation}\label{eq:neig}
    \begin{array}{ll}
        N_{\mathcal G}(k) & \mbox{as nodes adjacent to $k$ --- called \emph{open neighborhood},} \\
        N_{\mathcal G}[k] & \mbox{as nodes adjacent to $k$, including $k$ --- called \emph{closed neighborhood}.}
    \end{array}
\end{equation}

Another formalism that we use here is measurement
calculus on \emph{measurement patterns} (or patterns)~\cite{danos2007measurement} in order to describe processes
within the 1WQC scheme. A pattern comprises commands:
(i)~$N_j^{\theta}\coloneqq$ prepare qubit $j$ in state $\ket{+_{\theta}}$,
\begin{equation}
\ket{+_{\theta}}\coloneqq\frac{1}{\sqrt{2}}(\ket*{0}+e^{i{\theta}}\ket*{1}),
\end{equation}
(ii)~$E_{ij}\coloneqq$ apply CPHASE between qubits $i,j$.
(iii)~$M_j^\theta\coloneqq$ measure $j$ in the basis 
$\ket{\pm_\theta}$. Finally
(iv)~
$X^s_i,Z^s_i$ 
are Pauli corrections on qubit $i$ for signal $s$; that is, if $s=1$, the
corrections $X^1_i=X_i,Z^1_i=Z_i$ are done, and if $s=0$ no corrections are
done since $X^0_i=Z^0_i=\mathds{1}_i$. In addition, we denote $N_{Q}^\theta\coloneqq
\bigotimes_{i\in Q}N_{i}^{\theta_i}$ as preparing a set of qubits $Q$ accordingly.

Apart from those commands, we
introduce the following extra notations. First, $Z_i(\theta)$ signifies
rotation about $z$-axis on qubit $i$ with angle $\theta$; in particular 
\begin{equation}
    Z(\theta)\coloneqq\begin{pmatrix}1&0\\0&e^{i\theta}\end{pmatrix}.
\end{equation}
Second, given a graph
$\mathcal G=(V,E)$ with an ordering $>$ on the nodes, for any subgraph $\mathcal G[K]$ that $K\subset V$, 
\begin{equation}\label{eq:e_notation}
  E_K\coloneqq\prod_{(i,j)\in E(K)}E_{ij},\quad
  E_{iK} \coloneqq \prod_{k\in K} E_{ik}~\text{where}~i\not\in K,\quad
  E^{>}_{iK} \coloneqq\prod_{k\in K, k>i} E_{ik},
\end{equation}
where $E_{ij}$ is the entangling command (ii) in measurement calculus; $E_{ij}$ are all mutually commuting.

Note that Pauli corrections can be absorbed into measurement
angles~\cite{danos2007measurement}:
\begin{equation}\label{eq:correction}
    M_j^\theta X_j=M_j^{(-\theta)}\quad\quad M_j^\theta Z_j = M_j^{\theta+\pi},  
\end{equation}
where $-\theta$ and $\theta+\pi$ are understood to be evaluated modulo $2\pi$.

A set of angles $\theta=\{\theta_j\}_{j\in\mathbb N}$ specifies quantum
operations; angle $\theta_j$ can denote the parameter of a projective
measurement performed on node $j$ with projectors
\begin{equation}
    \{   ( \ketbra*{+_{\theta_j}}, \texttt{0} ),(\ketbra*{+_{\theta_j+\pi}}, \texttt{1})\},\label{basis}
\end{equation} 
where 
$\theta_j\in [0,2\pi)$.  The two projectors in (\Cref{basis}) will be reported as
outcomes \texttt{0} and \texttt{1} respectively. We also refer to such a 
measurement as ``measure in basis $\ket*{\pm_{\theta_j}}$''. An angle $\theta_j$ may be dependent on
signals $s_{<j}$, that is, measurement outcomes obtained previous to measuring
qubit $j$. It is inevitable that measurements introduce indeterminacies. Thus, adaptive measurements --- measurements which are depending on some
prior signals --- are performed.  The adaptive measurement in $\theta_j$ may be
$X_i$- or $Z_i$-dependent on a signal $i$.  If it is $X_i$-dependent,
$\theta_j$ is replaced by $(-1)^{s_i}\theta_j$; if it is $Z_i$-dependent,
$\theta_j$ is replaced by $\theta_j+s_i\pi$; these replacements are indicated 
in \Cref{eq:correction}. This dependency structure
is captured within the notion of {\em flow}.  It is worth mentioning that 
it is possible in 1WQC to perform measurements other than those in the
$xy$-plane, such as those in the $xz$- or
$yz$-plane~\cite{browne2007generalized}.  However, we consider only $xy$-plane
measurements here.

A \emph{flow} $f$ is a map from
the measured qubit to the prepared qubit. $f:O^c\rightarrow I^c$, where 
$A^c$ denotes the
complement of set $A$. More specifically for cluster states, $f(j)$ indicates the $X$
correction and $N_{\mathcal G}(f(j))$ indicate $Z$ corrections for the measured node
$j$.  By its definition, we will see that a flow induces a partial order
that covers $V$.  The state of an open graph $(\mathcal G,I,O)$ has flow if there exists
a map $f:O^c\rightarrow I^c$ together with a partial order $\succ$ over
nodes $V$ such that for all $j\in O^c$\cite{danos2006determinism,browne2007generalized}: 
\begin{equation} 
  \text{\bfseries (F0)}~(j,f(j))\in E,\quad
  \text{\bfseries (F1)}~f(j)\succ j,\quad
  \text{\bfseries (F2)}~\forall k\in N_{\mathcal G}(f(j))\setminus\left\{ j \right\}, k\succ j. 
\label{eqn:flow} 
\end{equation}
Hence, a 1WQC computation can be described with a set 
\begin{equation}
    \{(\mathcal G,I,O),f,\phi,\rho^{in}\},   
    \label{eq:description_1wqc_computation}
\end{equation}
where $(\mathcal G,I,O)$ denotes an open graph with flow $f$, $\phi$ signifies a set of measurement 
angles, and $\rho^{in}$ is a quantum state assigned to nodes $I$. 


Reference
\cite{danos2006determinism} also characterizes an interesting family of
patterns as follows. 

\begin{theorem}\label{thm:pattern}\cite{danos2006determinism}
Suppose the open graph state $(\mathcal G,I,O)$ has flow $(f,\succ )$,
then the pattern: 
\begin{equation} 
  \mathcal P_{f,\mathcal G,\succ,\vec{\theta}}\coloneqq \POP_{i\in
  O^c}\left (X^{s_i}_{f(i)}\left(\prod_{k\in N_{\mathcal G}(f(i)) \setminus \{i\}} Z_k^{s_i}\right)M_i^{\theta_i}\right)E_G N_{I^c}^0 \label{eq:pattern} \end{equation}
  is \emph{runnable}
and deterministic for all $\vec\theta$ and  $\vec s$. It realizes the 
isometry $\bigotimes_{i\in
O^c}\bra{+_{\theta_i}}_iE_{\mathcal G}N^0_{I^c}$, where $E_{\mathcal G}\coloneqq
\prod_{(i,j)\in \mathcal G} E_{ij}$.
\end{theorem}

A pattern is said to be runnable~\cite{danos2006determinism} if it satisfies
the following conditions: (R0) no command depends on an outcome not yet
measured, (R1) no command acts on a qubit already measured or not yet prepared
(except preparation commands), and (R2) a qubit $i$ is measured (prepared) if
and only if $i$ is not an output (input).

~\cref{thm:pattern} provides a sufficient condition in which the existence of
flow guarantees a deterministic computation. Notice that in \Cref{eq:pattern}
we introduce the notation $\POP$, signifying an ordered product with respect to
any given ordering $\succ$.  


\subsection{Abstract cryptography}
\label{sec:ac}

\emph{Abstract Cryptography} (AC)~\cite{maurer2011abstract} is a mathematical
framework that we use to model the security of our protocols. The AC framework
ensures \emph{composability}: a secure protocol in AC is
guaranteed to be secure when composing a larger cryptographic system. AC is a
constructive cryptography~\cite{maurer2012constructive}, \ie we are
constructing a resource from other weaker resources. To assess the security
in such a construction, AC implements an ``ideal-world real-world'' paradigm
in which the distance of resources is used as the security metric.  Simply put,
we want to construct an ideal resource (resources with desired security in the
ideal world) from a real resource (resources found in the real world), while
minimizing the distance between the two worlds (or systems). Perfect (or
unconditional) security is achieved when both systems are completely
indistinguishable.

AC uses a top-down approach: it starts from the highest-level possible of
abstraction, then proceeds downwards, introducing in each level the minimum
necessary specializations. The framework defines a \emph{system} as an abstract
object with interfaces. There are two types of systems: \emph{resources} and
\emph{converters}.  A \emph{resource} is a system with a set of interfaces
$\mathcal I=\{A, B, C,\dots\}$; each element $i\in\mathcal I$ is associated
with a player, which models how party $i$ can access the resource.  A
\emph{converter} is a system with two interfaces: an \emph{inside} interface
that is connected to a resource and the \emph{outside} interface that receives
inputs and gives outputs.

A \emph{(cryptographic) protocol} is a set of converters possessed by honest parties
$\{\pi_i\}$. Converters $\pi_i,\pi_j$ can be appended to a resource $\euscr R$
via their inside interfaces,\footnote{ Note that interfaces of $\euscr R$
contain $i$ and $j$.} forming a new resource $\pi_i\pi_j\euscr R$ with the
same set of interfaces.\footnote{The writing order is arbitrary, \ie
$\pi_i\pi_j\euscr R =\pi_j\pi_i\euscr R= \euscr R\pi_i\pi_j=\euscr R\pi_j\pi_i
= \pi_i\euscr R\pi_j =\pi_j\euscr R\pi_i$.} Resources and converters can be
instantiated with any mathematical object that follows the AC composition
defined in~\cite{maurer2011abstract}; in our case, we can instantiate them with
\emph{quantum combs}~\cite{chiribella2008quantum}. The quantum comb generalizes
quantum channels, mapping quantum circuits to quantum circuits rather than
quantum states to quantum states. Thus, here the compositions may be defined as
the operations of quantum combs as well.

A practical view of a protocol is seeing it as an
effort to obtain the ideal-world functionalities from the real-world
functionalities. 
Note that this framework does not capture the kind of
failure, the severity of failure, nor the cheating strategy.  Instead, we ask
if the defined ideal resource captures all the security features that we need;
otherwise, we define it differently.

Consider a three-party protocol $\pi$ with players Alice($\mathcal A$),
Oscar($\mathcal O$), and Bob($\mathcal B$), where Bob can be dishonest. There
are two security requirements that we want to achieve with the protocol $\pi$:
\emph{correctness} (or \emph{completeness}) and \emph{security} (or
\emph{soundness}). Correctness is a property that can be present only when all parties are honest (no
adversary present), defined in \Cref{def:correctness}.  Security is a property involving
the presence of an adversary (or adversaries), defined in
\Cref{def:security} for a cheating Bob.

\begin{definition}\cite{maurer2011abstract}
    Let $\pi=(\conv{A},\conv{O},\conv{B})$ be a protocol with no adversary and
    $\euscr {R,S}$ be resources. Protocol $\pi$ is $\varepsilon$-correct if it constructs an
    ideal resource within $\varepsilon$, viz., $\euscr R\xrightarrow{\pi,
    \varepsilon}\euscr S$, such that  \[d(\conv{ A}\conv{ O}\euscr{R}\conv{B},\euscr
    S)\leq\varepsilon.\]

    \label{def:correctness}
\end{definition}

\begin{definition}\cite{maurer2011abstract}
    Let $\pi=(\conv{ A},\conv{ O})$ be a protocol with an adversary $\mathcal B$
    and $\euscr {R,S'}$ be resources. Protocol $\pi$ is $\varepsilon$-secure, that is, it constructs an ideal 
    resource within $\varepsilon$, viz., $\euscr R\xrightarrow{\pi, \varepsilon}\euscr S'$, if there
    exists a converter $\simb$ (called \emph{simulator}), such that
        \[d(\conv{ A}\conv{ O}\euscr{R},\euscr{S'}\simb)\leq\varepsilon.\]
    \label{def:security}
\end{definition}

Here, $d$ signifies a pseudo-metric such that: $d(\euscr R,\euscr R)=0$,
$d(\euscr R,\euscr S)=d(\euscr S,\euscr R)$, and $d(\euscr R,\euscr S)\leq
d(\euscr R,\euscr T)+d(\euscr T,\euscr S)$.  If both definitions above are
fulfilled, we say protocol $\pi$ is $\varepsilon$-secure in producing tasks
defined in $\euscr S$, using resources $\euscr R$, and $\varepsilon$ is the
probability of failing.  For $\varepsilon=0$, we call the protocol $\pi$ 
unconditionally secure.

In \Cref{def:security}, with Bob being dishonest, an arbitrary system
(simulator $\simb$) is appended to $\euscr S'$ at the $\mathcal B$-interface in
order to make both systems ($\euscr S'$ and $\conv{A}\conv{O}\euscr R$)
comparable.  The system $\euscr S'\simb$ is also called a \emph{relaxation} of
$\euscr S'$~\cite{goldreich2006}, where the definition of simulator $\simb$ is
independent of Bob's cheating strategies.  A simulator assures that none of
these relaxations can be more useful to Bob than the ideal resource. As all
relaxations are defined in the ideal world, a real-world system is secure when
it is indistinguishable from at least one relaxation of the ideal system.

It now remains for us to practically specify ``distinguishing two resources''; here, we
use the notion of \emph{advantage}.\footnote{Advantage here defined as the
probability of guessing correctly minus guessing erroneously.} Two resources
$\euscr R$ and $\euscr S$ are indistinguishable if the rest of the world cannot tell whether it is interacting with $\euscr R$ or $\euscr S$; a
\emph{distinguisher} captures the rest of the world. One can think
of a distinguisher as a referee who has some access to and can freely interact
with an unknown system ($\euscr R$ or $\euscr S$). A distinguisher can read
outputs, give inputs, take the role of an adversary, generate an arbitrary
joint system, measure a joint system, and measure a
purification.\footnote{Involvement of purification systems is related to hiding
information explained in \cite{ambainis2009nonmalleable}.} The distinguisher
is then asked whether it is interacting with $\euscr R$ ($B=0$) or $\euscr S$
($B=1$). In this setting, the distance metric is called \emph{distinguishing
    advantage}.  Given $D$ as a random variable that signifies the
    distinguisher's guess, the distinguishing advantage is defined as 
    \begin{equation}
        \label{eq:distinguishing_advantage}
        \abs{\text{Pr}[D=0\mid B=0] - \text{Pr}[D=0\mid B=1]},
    \end{equation}
        which is the
        difference between guessing correctly and erroneously.  Perfect
        security is accomplished when the distinguishing advantage is zero.

\section{The BOQC}
\subsection{Pre-protocol}

This paper aims to provide a mechanism for a client to privately
delegate her oracular computation with the cooperation of a separate oracle
client. For this we propose our scheme, which we call Blind Oracular Quantum
Computation (BOQC). There are several BOQC protocols provided here, which
vary based on the used resources, \eg solid-state qubits, photonic qubits,
classical or quantum inputs, and based on whether the outputs are classical or quantum.

Consider the following illustration to give a picture of a situation in which a
BOQC will be used. Alice wants to run a private oracular quantum
algorithm,\footnote{Some examples of oracular algorithms can be found in
Ref~\cite{zoo}.} but she has neither a powerful quantum computer nor the
resources needed to construct her oracle function.  On the other hand, Oscar
has the information necessary to generate an oracle quantum circuit. Also, Bob
has a quantum computer, offered to clients as a service.  While building a
reliable quantum computer is very hard, in our idealized setting Alice and
Oscar will be motivated to use Bob's quantum computer.  However, they do not
wish to risk revealing information about their computation while running it on
Bob's quantum computer. Hence, they run the computation using the BOQC scheme.
Consider \Cref{fig:cartoon} that pictures this situation.

We assume that Alice and Oscar (the clients) possess the same level of quantum
power, and are always honest. The client-cooperation scheme is implemented with
minimal shared information, where Alice and Oscar are not communicating during
the protocol run; this is possible by the nature of an oracle in an algorithm,
which is a computation independent of the algorithm itself.  Since the
algorithm is run within the 1WQC scheme, Alice and Oscar represent their
algorithms as graphs and measurement angles. In particular, their preparations
are captured in \cref{def:prestep_boqc}, the pre-protocol steps, which are
performed before the protocol runs.

\begin{definition}\cite{gustiani}\label{def:prestep_boqc} Pre-protocol steps. Given that Oscar
agreed to provide oracle information for a quantum computation that Alice
wishes to run, the following steps are done via an authentic channel before
starting a BOQC protocol:

\begin{enumerate}[{\bfseries(B1)}]

    \item Alice and Oscar determine the size of bit string $b$; the string must
    be long enough to indicate all possible measurement angles $\phi_i$, $\psi_i\in\Omega$. The allowed set of angles is  $\Omega=\{ \pi k/2^{b-1}\}_{0\leq k < 2^b}$.

    \item Alice and Oscar join their graphs in the following way. Given that
    Alice's oracular quantum algorithm needs $k$ oracle queries, she marks each
    query as a black box on her graph $\mathcal A$. Given that Oscar's graph
    $\mathcal O$ is a graph with $k$ components, he sends Alice $\{\mathcal
    O,f_{\mathcal O}\}$. Alice joins their graphs with a connection $C$ by
    replacing each black box in $\mathcal A$ with a component of $\mathcal O$
    according to $C$; she obtains $\mathcal G=\mathcal A \cup_C \mathcal O$,
    and computes the total flow $(f,\succ)$ for the entire graph $\mathcal G$.
    The connection $C$ is valid when the resulting graph $\mathcal G$ has
    flow.\footnote{ See \Cref{exa:grover2} that shows the process of this step
    explicitly.}

    \item Alice determines a total ordering $>$ that is consistent with the
    partial ordering $\succ$.  This step is optional for the BOQC protocols
    (\Cref{pro:boqc,pro:boqcc}), but it is necessary for optimized protocols
    such as BOQCo protocols (\Cref{pro:boqco,pro:boqcocc}).

    \item Alice publicly informs Bob of $\{(\mathcal G, \tilde I, \tilde
    O),\vera,\vero, \succ, >, b \}$, where $\tilde I\subseteq I$ is a set of
    nodes assigned with quantum input $\rho^{in}$, $\tilde O\subseteq O$ is a
    set unmeasured nodes that will be sent to Alice, $\vera\equiv V(\mathcal
    A)$, and  $\vero\equiv V(\mathcal O)$.  If $\tilde I\neq \varnothing$ and
    $\tilde O\neq\varnothing$, then $\tilde I\cap\vero=\varnothing$ and $\tilde
    O\cap\vero=\varnothing$, respectively; simply put, Oscar does not give
    quantum input and does not receive quantum output.  The total order $>$ is
    given if step \textbf{(B3)} is executed.

\end{enumerate}
\end{definition}

Notice that we introduce $\tilde I$ and $\tilde O$ in addition to $I$ and $O$.
While variables $I$ and $O$ are required to define a 1WQC computation
(\Cref{eq:description_1wqc_computation}) such that one can describe the resulting computation
per \Cref{thm:pattern}, variables
$\tilde I$ and $\tilde O$ signify nodes whose quantum information is associated with
Alice. Set $\tilde I$ comprises nodes assigned with Alice's quantum input and
$\tilde O$ is a set of nodes whose state is received by Alice as the final output.
For a classical input $c=c_n\dots c_2c_1$, where $c_i\in\{0,1\}$, we can
consider it as quantum input with state $\prod_{i=1}^n\ketbra{c_i}$ in the
formalism used in \Cref{thm:pattern}.  Also, a classical output can be
considered as quantum output --- in the formalism used in \Cref{thm:pattern} ---
with a density of states whose matrix is diagonal.

\begin{example}\label{exa:grover2} 
    Alice wants to run a 2-qubit Grover algorithm and Oscar has the 4-element
    database. The algorithm comprises one oracle call and is implementable
    with circuit
    \[\scalebox{0.8}{\begin{quantikz}[row sep={2mm}]
        \lstick{$\ket 0_1$} &
        \gate[wires=2]{\mathscr A} &
        \gate[wires=2]{\mathscr O_f} &
        \gate[wires=2]{\mathscr D} & \meter{} \\
        \lstick{$\ket 0_2$} & & & & \meter{}
\end{quantikz}},\]
    where $\mathscr A$, $\mathscr O_f(\pi)$ and $\mathscr D(\pi)$ indicate
    preparation, oracle, and diffusion operators respectively. 
    The following steps show the steps in \Cref{def:prestep_boqc} to run a 2-qubit Grover algorithm;
    the algorithm (the graph state) is taken from~\cite{gustiani}.

First Alice and Oscar agree upon the bit string size, \eg $b=4$.  As Alice
does not need any quantum input and output, so she sets $\tilde I=\tilde
O=\varnothing$. Alice's graph state $\mathcal A$ contains a black box
indicating the oracle. Then, Oscar tells Alice his graph and flow $\{\mathcal
O,f_{\mathcal O}\}$, where $\mathcal O$ is a graph with one component.  The
following shows graphs of Alice and Oscar whose flows are indicated with
arrows.
\[\scalebox{0.9}{\begin{tikzpicture}[baseline=(current bounding box.center),font=\normalfont,scale=.7,auto,every node/.style={circle,inner
  sep=0pt,minimum size=12pt,draw=black,align=center}]
\coordinate (n3) at (1,1);
\coordinate (n4) at (1,0);
\coordinate (n5) at (2,1);
\coordinate (n6) at (2,0);
\coordinate (n) at (1,-.5);
\fill[black!70!white] (n) rectangle (2,1.4);
\node[] (n7) at (3,1){2};
\node[] (n8) at (3,0){1};
\node[] (n9) at (4,1){3};
\node[] (n10) at (4,0){4};
\draw[](n6) -- (n8);
\draw[](n9) -- (n10);
\draw[](n5) -- (n7);
\begin{scope}[decoration={markings,mark=at position 0.6 with {\arrow{>}}}]
\draw[postaction={decorate}](n7) -- (n9);
\draw[postaction={decorate}](n8) -- (n10);
\end{scope}
\begin{scope}[every node/.style={draw=none,align=center,inner sep=1pt}]
    \node[label=below:$(\mathscr{A})\mathscr{O}_f$] (o1) at (1.5,-1) {};
    \node[label=below:$\mathscr D$] (d1) at (3.5,-1) {};
  \draw[|-|] (1,-.7) -- (2,-.7);
  \draw[|-|] (3,-.7) -- (4,-.7);
\end{scope}
\end{tikzpicture}}\quad\eqqcolon\mathcal A
\hspace{2cm}
    \scalebox{0.9}{
    \begin{tikzpicture}[baseline=(current bounding box.center), font=\normalfont,scale=.7,auto,every node/.style={circle,inner
sep=0pt,minimum size=12pt,draw=black,align=center}]
\node[] (n2) at (0,1){6};
\node[] (n1) at (0,0){5};
\node[] (n3) at (1,1){7};
\node[] (n4) at (1,0){8};
\draw[](n2) -- (n1);
\begin{scope}[decoration={markings,mark=at position 0.6 with {\arrow{>}}}]
  \draw[postaction={decorate}](n2) -- (n3);
  \draw[postaction={decorate}](n1) -- (n4);
\end{scope}
\end{tikzpicture}}\quad\eqqcolon\mathcal O
\]
Here, $\mathcal A=(\{1,2,3,4\},\{(2,3),(1,4),(3,4)\})$ and $\mathcal
O=(\{5,6,7,8\},\{(5,6),(6,7),(5,8)\})$. Alice obtains the total graph $\mathcal
G=\mathcal A\cup_C\mathcal O$ with $C=\{(7,2),(8,1)\}$ and the total flow
$(f,\succ)$ as shown here:  
\[\scalebox{0.9}{\begin{tikzpicture}[baseline=(n10.center), font=\normalfont,scale=.7,auto,every node/.style={circle,inner
sep=0pt,minimum size=12pt,draw=black,align=center}]
\node[white, fill=black!70!white, draw=black!70!white ] (n3) at (1,1){6};
\node[white, fill=black!70!white, draw=black!70!white ] (n4) at (1,0){5};
\node[white, fill=black!70!white, draw=black!70!white ] (n5) at (2,1){7};
\node[white, fill=black!70!white, draw=black!70!white ] (n6) at (2,0){8};
\node[] (n7) at (3,1){2};
\node[] (n8) at (3,0){1};
\node[] (n9) at (4,1){3};
\node[] (n10) at (4,0){4};
\draw[](n3) -- (n4);
\draw[](n9) -- (n10);
\begin{scope}[decoration={markings,mark=at position 0.6 with {\arrow{>}}}]
\draw[postaction={decorate}](n3) -- (n5);
\draw[postaction={decorate}](n4) -- (n6);
\draw[postaction={decorate}](n5) -- (n7);
\draw[postaction={decorate}](n6) -- (n8);
\draw[postaction={decorate}](n7) -- (n9);
\draw[postaction={decorate}](n8) -- (n10);
\end{scope}
\node[draw=none,font=\normalfont] at (2.5,-.6){$\mathcal G$};
\end{tikzpicture}}
\quad
\scalebox{0.9}{
\begin{tikzpicture}[baseline=(n10.center), font=\small\normalfont,scale=.7,auto,every node/.style={circle,inner
sep=0pt,minimum size=12pt,draw=black,align=center}]
\node[white, fill=black!70!white, draw=black!70!white,label={[gray,anchor=east]1}] (n3) at (1,1){6};
\node[white, fill=black!70!white, draw=black!70!white,label={[gray,anchor=east]1}] (n4) at (1,0){5};
\node[white, fill=black!70!white, draw=black!70!white,label={[gray,anchor=east]2}] (n5) at (2,1){7};
\node[white, fill=black!70!white, draw=black!70!white,label={[gray,anchor=east]2}] (n6) at (2,0){8};
\node[label={[gray,anchor=east]3}] (n7) at (3,1){2};
\node[label={[gray,anchor=east]3}] (n8) at (3,0){1};
\node[label={[gray,anchor=east]4}] (n9) at (4,1){3};
\node[label={[gray,anchor=east]5}] (n10) at (4,0){4};
\draw[](n3) -- (n4);
\draw[](n9) -- (n10);
\begin{scope}[decoration={markings,mark=at position 0.6 with {\arrow{>}}}]
\draw[postaction={decorate}](n3) -- (n5);
\draw[postaction={decorate}](n4) -- (n6);
\draw[postaction={decorate}](n5) -- (n7);
\draw[postaction={decorate}](n6) -- (n8);
\draw[postaction={decorate}](n7) -- (n9);
\draw[postaction={decorate}](n8) -- (n10);
\end{scope}
\node[draw=none,font=\normalfont] at (2.5,-.6){$\mathcal G,\succ$};
\end{tikzpicture}},\]
with $I=\{5,6\}$,
$\tilde I=\varnothing$, $O=\{3,4\}$, $\tilde O=\varnothing$,
$\vera=\{1,2,3,4\}$, $\vero=\{5,6,7,8\}$, partial ordering
$\{5,6\}\succ\{7,8\}\succ\{1,2\}\succ\{3\}\succ\{4\}$, and a total ordering, \eg
$5<6<7<8<1<2<3<4$.  Finally, Alice publicly tells Bob $\{(\mathcal
G,\varnothing,\varnothing),\vera,\vero,>,\succ,4\}$.

In this example, neither Alice nor Oscar provide classical or quantum input, and
$I\subset \vero$; this is consistent with \Cref{def:prestep_boqc} since $\tilde
I=\varnothing$. Instead, here the input is implicit, \ie two zeros $\ketbra{00}$.
%
%
\end{example}
\subsection{The protocol}
\label{sec:boqc_protocol}

The BOQC protocol involves three players: Alice, Bob, and Oscar, indicated with
$\mathcal A$, $\mathcal B$, and $\mathcal O$, respectively.  In the language of
AC, we denote the protocol as $\pi_{boqc}=(\conv{A},\pi_{\mathcal
B},\conv{O})$ --- with an honest Bob --- which uses a real-world resource $\euscr
R$. The resource $\euscr R$ comprises three channels: a secure key channel
between Alice and Oscar, a two-way insecure classical channel between each
client and Bob, and one-way quantum channels between each client and Bob. If
Alice expects quantum outputs, a two-way quantum channel between Alice and Bob is
needed.

\begin{protocol}[!h]
    \caption{BOQC ($\pi_{boqc}=\{\conv{A},\conv{B},\conv{O}\}$)}
  \label{pro:boqc}
  \begin{algorithmic}[1]

      \Statex{\hspace{-2em}\bfseries Alice's input: $\{(\mathcal G, I, O),f,\phi,\rho^{in}_{\mathcal A}\}$}\Comment{$\tilde I=I$ and $\tilde O=O$}
  \Statex{\hspace{-2em}\bfseries Oscar's input: $\{\psi\}$}
  \Statex{\hspace{-2em}\bfseries Alice's output for an honest Bob:
  $\rho^{out}_{\mathcal A}=\mathcal E(\rho^{in}_{\mathcal A})$}

  \Statex{\itshape Assumptions and conventions: 
          \begin{enumerate}[(I)]
              \item Alice ($\mathcal A$) and Oscar ($\mathcal O$) have
                  performed pre-protocol steps in \Cref{def:prestep_boqc}; 
                  Bob knows $\{(\mathcal G,\tilde I,\tilde O),\vera,\vero,\succ,>,b\}$. 
                  Here, we set $\tilde I=I$ and $\tilde O=O$. Recall $\tilde O \cap \vero=\varnothing$
                  (quantum outputs are held by Alice) and $\Omega=\{\frac{\pi k}{2^{b-1}}\}_{0\leq k<2^b}$.

      \item $s_{\invf{i}}=0,\forall i\in I$ and $t_i=0,\forall i\in I^c$. 
          
      \item $\invf{i}\equiv f^{-1}(i)$,
          $z(i)\coloneqq\bigoplus_{k\prec i, i\in \neig{f(k)}}s_k$, and  $t(i)\coloneqq\bigoplus_{k\in I, i\in N_{\mathcal G}(k)}t_k$.
          \end{enumerate}}

  \Statex{\hspace{-2em}\bfseries \circled{0} Pre-preparation}

  \State{Alice and Oscar receive keys $r,t$ via a secure key channel,
  where $r_i\in\{0,1\},~i\in O^c$ and $t_j\in\{0,1\},~j\in I$.}

    \Statex{\hspace{-2em}\bfseries\centering \circled{1} State
    preparation}

    \For{\label{ln:boqc_loop_prepare}$i \in V \setminus O$ following partial ordering $\succ$}\Comment{it may follow ordering $>$}
    \If{$i\in \vera $} 
    \State{Alice chooses $\alpha_i\in\Omega$ at random.}
    \If{$i \in I$}

    \State{\label{ln:boqc_start_I}Alice applies $Z_i(\alpha_i)X^{t_i}_i$ to $\tr_{I\setminus i}[\rho^{in}_{\mathcal A}]$  and sends it to Bob.}
    \State {Alice updates angles: \vspace{-1em}
                    \begin{align*}
                    \phi_i =& (-1)^{t_i}\phi_i \\
                    \phi_j=&\phi_j+t_i \pi,~\forall j\in N_{\mathcal G}(i)\cap\vera.
        \end{align*}
  }
    \State{Oscar updates angles 
    $\psi_j=\psi_j+t_i \pi,~\forall j\in N_{\mathcal G}(i)\cap\vero.$
}\label{ln:boqc_end_I}
  \Else
      \State\parbox[t]{0.89\linewidth}{Alice prepares
        $\ket{+_{\alpha_i}}_i$ and sends it to Bob.}
      \EndIf
    \ElsIf { $i \in \vero$}
      \State{Oscar prepares $\ket{+_{\beta_i}}_i$ and sends it to Bob, where $\beta_i\in\Omega$ is chosen at random.}
    \EndIf
    \EndFor

    \State{For all $i\in O$, Bob prepares $\ket{+}_i$.}\label{ln:boqc_bob_prepare_out}

    \Statex{\hspace{-2em}\bfseries\centering \circled{2} Graph state formation}
    \State{Bob applies entangling operator $E_{\mathcal G}$ defined in \Cref{eq:e_notation}.}

    \Statex{\hspace{-2em}\bfseries \circled{3} Classical interaction and measurement}


    \For{$i \in V\setminus O$ which follows partial ordering $\succ$}\label{ln:boqc_loop3}

    \If{$i\in \vera$} 

    \State{\label{ln:a_phi}Alice computes $\phi'_i=(-1)^{s_{\invf{i}}}\phi_i+z(i)\pi$.}
    \State{\label{ln:a_delta}Alice computes $\delta_i\coloneqq \phi'_i + \pi r_i + \alpha_i$, and sends Bob $\delta_i$.} 

    \ElsIf { $i \in \vero$}
    \State{\label{ln:o_psi}Oscar computes $\psi'_i=(-1)^{s_{\invf{i}}}\psi_i+ z(i)\pi$.}
    \State{\label{ln:o_delta}Oscar computes $\delta_i\coloneqq \psi'_i + \pi r_i + \beta_i$, and sends Bob $\delta_i$.} 
    \EndIf

    \State{Bob measures qubit $i$ in $\ket{\pm_{\delta_i}}$ basis, then sends Alice and Oscar the outcome $\tilde s_i$.}

    \State{Alice and Oscar set $s_i=\tilde s_i\oplus r_i$.}\label{ln:flip}
    \EndFor

    \Statex{\hspace{-2em}\bfseries\centering \circled{4} Output
    transmission and correction}

    \State{Bob sends Alice output qubits $\rho^{out}_{\mathcal B}$ (all qubits $i\in O$).} 

    \State{Alice corrects the final output 
         $P(\rho^{out}_{\mathcal B})\eqqcolon\rho^{out}_{\mathcal A}$, where
         $P\equiv\bigotimes_{i\in O}X^{s_{\invf{i}}+t_i}_i
     Z_i^{z(i)+t(i)}$.} 

 \end{algorithmic}
\end{protocol}

\footnotetext[1]{It means the loop goes through all nodes in $V\setminus O$
following a total ordering $>$ that is consistent with $\succ$; the total ordering $>$ is 
implicit there.} 

\footnotetext[2]{Operator $\tr_j$ is a partial trace, where subsystem $j$ is traced-out; thus 
$\tr_{I\setminus i}[\rho^{in}_{\mathcal A}]$ indicates that subsystem $i$ remains, where $i\in I$.}

The BOQC protocol is provided in \Cref{pro:boqc}, in a clear-cut style with
explicit adaptive measurements. For generality, \Cref{pro:boqc} admits the case
in which Alice requires quantum input and quantum output; integrating classical
inputs and outputs to the protocol is straightforward, and it is also
discussed.  Note that symbol $\oplus$ is defined as modulo 2 addition.

Every usage of a resource in $\euscr R$ --- interactions with communication
channels --- is depicted in \Cref{fig:boqc_communication}, where a circled
number corresponds to the indicated part in \Cref{pro:boqc}.
\Cref{fig:boqc_communication} shows that Alice and Oscar alternately take over
the computation (transmitting qubits and giving commands with measurement
angles) without communicating with each other, apart from sharing a key in the beginning.

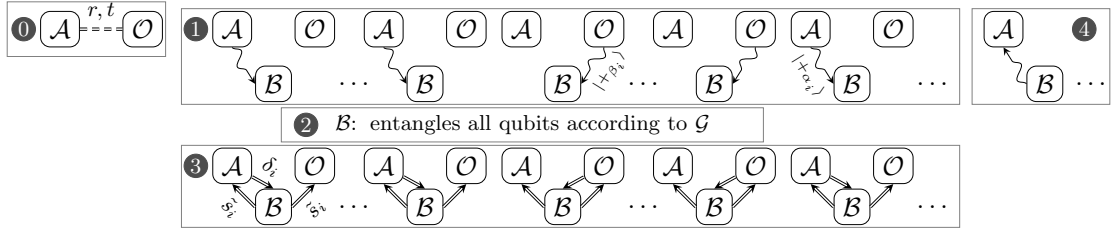
\begin{figure}[htbp]
  \centering
  \begin{minipage}{\textwidth}
    \resizebox{\columnwidth}{!}{
\begin{tikzpicture}

  \begin{scope}[>=stealth, every edge/.append
    style={decorate,decoration={snake, post length=.6mm}} ,every node/.append style={sloped,anchor=center}]

  \node (a11) at (-.5,  2.6) [abo] {$\mathcal A$};
  \node (o11) at (.7,2.6) [abo] {$\mathcal O$};
  \draw [-, double, densely dashed](a11.east) -- node[midway,above]{\footnotesize$r,t$}(o11.west);

  \node (a21) at (2,  2.6)[abo] {$\mathcal A$};
  \node (o21) at (3.2,2.6)[abo] {$\mathcal O$};
  \node (b21) at (2.6,1.8)[abo] {$\mathcal B$};
  \path [->] (a21.south) edge node [below, text width=.6cm]{}(b21.west);

  \node (d1) at (3.8,1.8) {\dots};

  \node (a31) at (4.2,  2.6)[abo] {$\mathcal A$};
  \node (o31) at (5.4,2.6)[abo] {$\mathcal O$};
  \node (b31) at (4.8,1.8)[abo] {$\mathcal B$};
  \path [->] (a31.south) edge node {}(b31.west);

  \node (a41) at (6.2, 2.6)[abo] {$\mathcal A$};
  \node (o41) at (7.4, 2.6)[abo] {$\mathcal O$};
  \node (b41) at (6.8, 1.8)[abo] {$\mathcal B$};
  \path [->] (o41.south) edge node [below, text
  width=.6cm]{\tiny$\ket*{+_{\beta_i}}$}(b41.east);; 

  \node (d2) at  (8,1.8) {\dots};

  \node (a51) at (8.4, 2.6)[abo] {$\mathcal A$};
  \node (o51) at (9.6, 2.6)[abo] {$\mathcal O$};
  \node (b51) at (9,   1.8)[abo] {$\mathcal B$};
  \path [->] (o51.south) edge node {}(b51.east);

  \node (a61) at (10.4, 2.6)[abo] {$\mathcal A$};
  \node (o61) at (11.6, 2.6)[abo] {$\mathcal O$};
  \node (b61) at (11  , 1.8)[abo] {$\mathcal B$};
  \path [->] (a61.south) edge node[midway,below] {\tiny$\ket*{+_{\alpha_i}}$} (b61.west);

  \node (d3) at (12.2,1.8) {\dots};

  \node (a71) at (13.2, 2.6)[abo] {$\mathcal A$};
  \node (b71) at (13.8, 1.8)[abo] {$\mathcal B$};

  \node (d3) at (14.5,1.8) {\dots};

  \path [->] (b71.west) edge node {}(a71.south);
\end{scope}

\node [text width=7cm](e1) at (7,1.2) {\footnotesize $\mathcal B$: entangles all qubits according to $\mathcal G$};

  \begin{scope}[font=\footnotesize,auto,every node/.style={circle,inner
    sep=0pt,minimum  size=6pt, draw=black,text width=6pt,align=center}]
    \node [draw=none](n0) at (-1.1, 2.6) {\circled{0}};
    \node [draw=none](n1) at (1.4, 2.6){\circled{1}};
    \node [draw=none](n2) at (3,1.2) {\circled{2}};
    \node [draw=none](n3) at (1.4,.6) {\circled{3}};
    \node [draw=none](n4) at (14.3,2.6) {\circled{4}};
  \end{scope}
  \begin{scope}[every node/.append style={draw, gray, shape=rectangle, minimum width=2.3cm, minimum height=.5cm, anchor=center}]
      \node(box0)[minimum height=.8cm] at ([xshift=2em]n0.east){};
      \node[minimum width=11.3cm, minimum height=1.4cm](box1) at ([xshift=5.3cm,yshift=-.4cm]n1.east){};
      \node(box2)[minimum width=7cm] at ([xshift=3cm]n2.east){};
      \node[minimum width=11.3cm, minimum height=1.2cm](box3) at ([xshift=5.3cm,yshift=-.3cm]n3.east){};
      \node(box4)[minimum height=1.4cm, minimum width=2cm] at ([xshift=-2em,yshift=-.4cm]n4.east){};
  \end{scope}

\begin{scope}[>=stealth,every edge/.append style={double}, every node/.append style={sloped,anchor=center}]
  \node (a1) at (2,   .6)[abo] {$\mathcal A$};
  \node (o1) at (3.2, .6)[abo] {$\mathcal O$};
  \node (b1) at (2.6, 0)[abo] {$\mathcal B$};
  \path [->](a1) edge node [above]{\footnotesize$\delta_i$} (b1.north);
 \path [->] (b1.west) edge node [below]{\footnotesize$\tilde s_i$} (a1.south);
 \path [->] (b1.east) edge node [below]{\footnotesize$\tilde s_i$} (o1.south);

 \node (d1) at (3.8,0) {\dots};

 \node (a2) at (4.2,.6)[abo] {$\mathcal A$};
 \node (o2) at (5.4,.6)[abo] {$\mathcal O$};
 \node (b2) at (4.8,0)[abo] {$\mathcal B$};
 \path [->] (a2) edge node {} (b2.north);
 \path [->] (b2.west) edge node {} (a2.south);
 \path [->] (b2.east) edge node {} (o2.south);

 \node (a3) at (6.2,.6)[abo] {$\mathcal A$};
 \node (o3) at (7.4,.6)[abo] {$\mathcal O$};
 \node (b3) at (6.8,0)[abo] {$\mathcal B$};
 \path [->] (o3) edge node [above]{} (b3.north);
 \path [->] (b3.west) edge node [below]{} (a3.south);
 \path [->] (b3.east) edge node [below]{} (o3.south);

 \node (d2) at (8,0) {\dots};

 \node (a4) at (8.4,.6)[abo] {$\mathcal A$};
 \node (o4) at (9.6,.6)[abo] {$\mathcal O$};
 \node (b4) at (9, 0)[abo] {$\mathcal B$};
 \path [->] (o4) edge node [above]{} (b4.north);
 \path [->] (b4.west) edge node [below]{} (a4.south);
 \path [->] (b4.east) edge node [below]{} (o4.south);

 \node (a5) at (10.4,.6)[abo] {$\mathcal A$};
 \node (o5) at (11.6, .6)[abo] {$\mathcal O$};
 \node (b5) at (11,  0)[abo] {$\mathcal B$};
 \path [->] (a5) edge node [above]{} (b5.north);
 \path [->] (b5.west) edge node [below]{} (a5.south);
 \path [->] (b5.east) edge node [below]{} (o5.south);

 \node (d3) at (12.2,0) {\dots};


\end{scope}

\end{tikzpicture}

  } \end{minipage} 
  
  \caption{ Interactions within BOQC scheme; these also describe the
    access to the inside interface of resource $\euscr R$. Initials $\mathcal A$,
    $\mathcal B$, and $\mathcal O$ indicate the interface of Alice, Bob, and Oscar
    respectively. The double dashed line
    indicates a secure key channel; the double lines indicate 
    classical channels, with arrow indicating direction of transmission; the wavy lines indicate
    quantum channels.  The circled numbers \circled{0}--\circled{4} correspond to the steps shown in
    \Cref{pro:boqc}. 
      } \label{fig:boqc_communication}
    \end{figure}%

In \Cref{pro:boqc} Bob receives input from Alice as $\{(\mathcal
G,I,O),f,\phi,\rho^{in}_{\mathcal A}\}$ and input from Oscar as $\{\psi\}$, where
$\mathcal G$ is the total graph with flow $f$, $I$ is a set of input nodes that will be
assigned with input state $\rho^{in}_{\mathcal A}$, $O$ is a set of output
nodes, $\phi$ is a set of measurement angles of Alice's nodes, and $\psi$ is a
set of measurement angles of Oscar's nodes. Prior to the protocol, we assume
pre-protocol steps in \Cref{def:prestep_boqc} have been successfully done.

\Cref{pro:boqc} comprises five steps: \circled{0} pre-preparation, \circled{1}
state preparation, \circled{2} graph state formation, \circled{3} classical
interaction and measurements, and \circled{4} output transmission and
correction. The interactions of players via a channel at each step are
depicted in \Cref{fig:boqc_communication}.  The protocol is initiated by
establishing a symmetric key between Alice and Oscar via a secure key channel
--- step \circled{0}. This step allows them to privately delegate their joint
computation: in particular, it allows them to know the actual measurement
outcomes so that their adaptive measurements can be computed independently.
Step \circled{1}--\circled{4} comprises the delegated computation process:
\circled{1} the clients send Bob the encrypted qubits, \circled{2} Bob
entangles the received qubits, \circled{3} the clients communicate to Bob the
measurement angles, and \circled{4} Bob sends the output qubits if necessary,
\ie quantum outputs.


The following lemma (\Cref{lem:correction}) shows that, for any
computation, the obtained pattern in \Cref{pro:boqc} without randomness is
identical to the obtained pattern in the 1WQC scheme.  The proof of
\Cref{lem:correction} is available in \Cref{app:proofs}.

\begin{restatable}{lemma}{lemcorrection}\cite{gustiani}\label{lem:correction}
    Suppose the open graph state $(\mathcal G,I,O)$ has flow $(f,\succ)$,
    then the following patterns $\mathfrak P_1,\mathfrak P_2$ are identical $\forall\phi_i$: 
  \begin{align}
    \mathfrak P_1\coloneqq& \POP_{i\in O^c} ( 
    X_{f(i)}^{s_i}
    \prod_{k\in\neig{f(i)}\setminus \{i\}} Z_k^{s_i}  
    M_i^{\phi_i}
)E_{\mathcal G}N_{I^c}^0 \label{eq:p1}
    \\ 
    \mathfrak P_2\coloneqq&
    \bigotimes_{j\in O}
    X_j^{s_{\invf{j}}}
    Z_j^{z(j)}
    \,\POP_{i\in O^c} ( 
    M_i^{\phi_i}
    X_i^{s_{\invf{i}}}
    Z_i^{z(i)}
)E_{\mathcal G}N_{I^c}^0,\label{eq:p2}
  \end{align}
  where $z(i)\coloneqq\bigoplus_{k\prec i,i\in \neig{f(k)}}s_k$, $\invf{i}\equiv f^{-1}(i)$, and 
  $\invf{i}=0$ for all $i\in I$. 
\end{restatable}

Patterns $\mathcal P_1$ and $\mathcal P_2$ give different points of view in
writing the corrections: Pattern $\mathcal P_1$ shows the obtained correction
from measuring $i$, while pattern $\mathcal P_2$ shows the corrections that
can be done before measuring $i$.\footnote{The significance of this point of
view is in our BOQCo protocol, to be shown later in \Cref{sec:boqco}. In this
point of view, the causality of BOQCo becomes apparent.} Those points of
view were first introduced in \cite{fitzsimons2017unconditionally}. See
\Cref{tab:corrections} for an illustration.		

\begin{table}[bth]
        \centering
        \scalebox{0.9}{
            \renewcommand*{\arraystretch}{0.9}
        \begin{tabular}{|c|c|c|c|c|c|}
\hline
node $i$ & $f(i)$ & $\neig{f(i)}$ &$\neig{f(i)}\setminus \{i\}$&$\invf{i}$ &  $z(i)$ \\  
\hline
1  &  4&   $ \{1,3\}$  &$\{3\}$&8 & $s_5$\\
2  &  3&   $ \{2,4\}$  &$\{4\}$&7 & $s_6$\\
3  &  $-$& $    -   $  &$  -  $&2 & $s_1\oplus s_7$ \\
4  &  $-$& $    -   $  &$  -  $&1 & $s_2\oplus s_8$\\
5  &  8&   $ \{1,5\}$  &$\{1\}$&$-$& $-$\\
6  &  7&   $ \{2,6\}$  &$\{2\}$&$-$& $-$\\
7  &  2&   $ \{3,7\}$  &$\{3\}$&6 & $-$\\
8  &  1&   $ \{4,8\}$  &$\{4\}$&5 & $-$\\
\hline
\end{tabular}}
\hspace{1em}
\scalebox{0.9}{
\begin{tikzpicture}[baseline=(n10.center), font=\small\normalfont,scale=.7,auto,every node/.style={circle,inner
sep=0pt,minimum size=12pt,draw=black,align=center}]
\node[white, fill=black!70!white, draw=black!70!white,label={[gray,anchor=east]1}] (n3) at (1,1){6};
\node[white, fill=black!70!white, draw=black!70!white,label={[gray,anchor=east]1}] (n4) at (1,0){5};
\node[white, fill=black!70!white, draw=black!70!white,label={[gray,anchor=east]2}] (n5) at (2,1){7};
\node[white, fill=black!70!white, draw=black!70!white,label={[gray,anchor=east]2}] (n6) at (2,0){8};
\node[label={[gray,anchor=east]3}] (n7) at (3,1){2};
\node[label={[gray,anchor=east]3}] (n8) at (3,0){1};
\node[label={[gray,anchor=east]4}] (n9) at (4,1){3};
\node[label={[gray,anchor=east]5}] (n10) at (4,0){4};
\draw[](n3) -- (n4);
\draw[](n9) -- (n10);
\begin{scope}[decoration={markings,mark=at position 0.6 with {\arrow{>}}}]
\draw[postaction={decorate}](n3) -- (n5);
\draw[postaction={decorate}](n4) -- (n6);
\draw[postaction={decorate}](n5) -- (n7);
\draw[postaction={decorate}](n6) -- (n8);
\draw[postaction={decorate}](n7) -- (n9);
\draw[postaction={decorate}](n8) -- (n10);
\end{scope}
\node[draw=none,font=\normalfont] at (2.5,-.6){$\mathcal G,\succ$};
\end{tikzpicture}}

\caption{\label{tab:corrections}Correction terms in \Cref{lem:correction} with
    graph $(\mathcal G,I,O)$ from \Cref{exa:grover2}.  In pattern $\mathcal P_1$, upon
    measuring $i$, $X$-correction goes to node $f(i)$ and $Z$-corrections go to
    nodes $\neig{f(i)}\setminus\{i\}$. In pattern $\mathcal P_2$, before
    measuring $i$, $X$-correction is performed based on measurement $\invf{i}$ and
    $Z$-correction based on measurements $z(i)$.  }
\end{table} 

Using \Cref{lem:correction}, we obtain \Cref{thm:pattern}, which states that an
algorithm run within the BOQC implements the same map as when the algorithm is
run directly within the 1WQC, without requiring Alice and Oscar to share any of
their secrets.  The proof of \Cref{thm:correctness_boqc} is available in
\Cref{app:proofs}.

\begin{restatable}{theorem}{thmcorrectnessboqc}\cite{gustiani}
    The BOQC protocol $\pi_{boqc}$ defined in \Cref{pro:boqc} delegates a
    computation with the isometry defined in \Cref{thm:pattern}
    for the same computation, without requiring Alice
    and Oscar to communicate their computations to each other.
  \label{thm:correctness_boqc}
\end{restatable}

\subsection{The quantum power}
\label{sec:quantum_power}

Like UBQC, BOQC is a protocol that involves a clear separation of quantum power
between client and server, where a client is only capable of producing and transmitting
quantum states of the form $\ket{+_{\theta}}$, while the server is presumed to
posses unlimited quantum power.  While this formulation is exactly true for classical input
and output ($\tilde I=\tilde O=\varnothing$), a client needs higher quantum
power for quantum input and output, \ie $\tilde I\neq \varnothing$ or $\tilde
O\neq \varnothing$.  \Cref{tab:qpower} summarizes the minimal quantum power
requirement of running BOQC protocols with various input-output types. 

\begin{table}[th]
\centering
\scalebox{0.85}{
    \begin{tabular}{ |c |c | p{6cm}| p{6cm}|c|}
        \hline
        $I$ & $O$ & \multicolumn{1}{c|}{\textbf{Client} (Alice)} & \multicolumn{1}{c|}{\textbf{Server} (Bob)}
        & \multicolumn{1}{c|}{\textbf{Resource}} \\
    \hline
    $c$ & $c$ &  \textbf{(C1)} creates $\ket{+_\theta}$, then transmits it to Bob & 
    \textbf{(S1)} receives $\ket{+_\theta}$, performs CPHASE gates, 
    and measures qubits in the $xy$-plane bases  &  $\euscr{K,C,Q}$  \\
    $c$&$q$& \textbf{(C1)}, \textbf{(C2)} receives $\rho^{out}_{\mathcal B}$ from Bob, then performs Pauli correction on it&
    \textbf{(S1)}, \textbf{(S2)} create $\ket + \forall i\in O$ and sends Alice the final outputs $\rho^{out}_{\mathcal B}$
       & $\euscr{K,C,Q,Q}2$ 
    \\
     $q$&$c$& \textbf{(C1)}, \textbf{(C3)} creates quantum input $\rho^{in}_{\mathcal A}$, 
     performs a quantum one-time pad (\eg line~\ref{ln:boqc_start_I} in \Cref{pro:boqc}) then transmits it to Bob &
    \textbf{(S1)}, \textbf{(S3)} receives input states (arbitrary) from Alice 
    & $\euscr{K,C,Q}$ 
    \\
    $q$&$q$&  \textbf{(C1)}, \textbf{(C2)}, \textbf{(C3)}  & 
    \textbf{(S1)}, \textbf{(S2)}, \textbf{(S3)} & $\euscr{K,C,Q,Q}2$ 
    \\
   \hline
\end{tabular}}
\caption[Minimal quantum power requirements in the BOQC]{\cite{gustiani}Minimal quantum power 
    requirements of Alice and Bob to run a BOQC protocol. 
    Oscar's requirement remains \textbf{(C1)} for all cases.
    Initial ``$c$'' indicates ``entirely classical'' and initial ``$q$''
indicates ``entirely quantum''. Resources $\euscr{K,C,Q}$, and $\euscr{Q}2$ are
respectively signifying a secure key channel between Alice and Oscar, 
an insecure classical channel between each client and server,
a one-way quantum channel (each client to server), and a two-way quantum channel between Alice and Bob.
Other notations follow the notations in \Cref{pro:boqc}.}
\label{tab:qpower}
\end{table}

\Cref{pro:boqc} admits the general case, \ie $\tilde I=I$ and $\tilde O=O$. As
shown in \Cref{tab:qpower}, it has the highest requirement among all the cases. A
few adjustments from \Cref{pro:boqc} are needed if $\tilde I\subset I$ or
$\tilde O \subset O$.

Given an entirely classical input case $\tilde I=\varnothing$, \eg
Alice's input is a binary string $c=c_1c_2\dots c_n$, where $c_i\in\{0,1\}$,
lines \ref{ln:boqc_start_I}--\ref{ln:boqc_end_I} in \Cref{pro:boqc} turns into a
single line:
\[\text{``\emph{Alice prepares} $\ket{+_{\alpha_i + \pi c_i}}$'';}\]
recall that $Z\ket{+}=\ket{+_{\pi}}$.
Since the quantum one-time pad is unnecessary, the random 
string $t$ is omitted (or setting $t_i=0,\forall i$). Thus, requirements 
\textbf{(C3)} and \textbf{(S3)} vanish.

Now, when $\tilde O=\varnothing$ (entirely classical output), measurements will be performed on all
qubits $i\in V$ and Bob does not need to prepare state $\ket +$ himself
nor to send Alice the final outcome $\rho_{\mathcal B}^{out}$. 
The first removes requirement \textbf{(C2)}, which replaces the loop
in line~\ref{ln:boqc_loop3} with 
\[\text{``\textbf{for}  $i\in V$ which follows partial ordering $\succ$.''}\]
The latter eliminates requirement $\textbf{(S2)}$, which removes 
line~\ref{ln:boqc_bob_prepare_out} and removes step \circled{4} entirely.

Therefore, the lowest quantum power demand occurs for the entirely classical
input and output case, \ie $\tilde I=\tilde O=\varnothing$, requiring only
$\textbf{(C1)}$ and $\textbf{(S1)}$. We provide an explicit BOQC protocol for
entirely classical input and output in \Cref{pro:boqcc}, \Cref{app:protocols}.

\section{BOQC optimized: BOQC on solid-state qubits}\label{sec:boqco}
The 1WQC --- as well as BOQC --- can efficiently perform computations on
memoryless quantum computers, such as photonic qubits, which is shown by
successful experimental demonstrations on linear optics quantum computers: 1-
and 2-qubit gates~\cite{walther2005experimental}, 2-qubit Grover's
algorithm~\cite{chen2007experimental}, Deutsch's
algorithm~\cite{tame2007experimental}, blind quantum
computing~\cite{barz2012demonstration}, and verification of quantum
computations~\cite{barz2013experimental}. However, extending the experiments to 
perform more complex computations is very hard since the individual qubit control is
tricky in memoryless qubit systems.
Individual-qubit control on solid-state
systems are more promising, but scalability remains challenging.  This problem
motivates us to come up with an ``optimized'' version of the BOQC, which we
will call BOQCo (\emph{Blind Oracular Quantum Computation-Optimized}).

BOQCo allows us to perform BOQC algorithms with a minimal number of physical
(solid-state) qubits.  BOQCo is runnable on an appropriate
platform\footnote{These are a quantum-network platform in which Bob owns a
solid-state quantum system, \eg NV-center and trapped-ions.} whose qubits
possess permanence and can be rapidly re-initialized.

We prepare the graph in parts to minimize the number of physical qubits;
the qubits are initialized only when needed. We call such a strategy
\emph{lazy 1WQC} in which the server only needs to prepare the closed
neighborhood of the qubit about to be measured.\footnote{The ``lazy'' name is
inspired from a computational programming paradigm called \emph{lazy evaluation}.
Lazy evaluation means that the evaluation of an expression is delayed until the
value is needed~\cite{watt2004programming}} Note that such a graph preparation
has been introduced in~\cite{housmand2018} but restricted to graphs with
the same number of inputs and outputs ($\abs{I}=\abs{O}$); here, we extend it
to arbitrary graphs with flow.

\subsection{Lazy 1WQC computation}\label{sec:lazy}

The lazy 1WQC is a 1WQC-computation scheme that allows one to prepare parts of
the graph state as needed such that the number of physical qubits is minimal.
The lazy 1WQC scheme is shown in \Cref{alg:lazy} with input
\begin{equation}\{(\mathcal G^>,I, O),f,\phi,\rho^{in}\},\end{equation}
where $(\mathcal G^>,I,O)$ is 
an open graph with total
ordering $>$ that follows the flow $f$, $\phi$ is the set of measurement angles, and $\rho^{in}$
is the state assigned to $I$. Note that to describe a lazy 1WQC, one needs an additional parameter, total ordering $>$, compared with the description of a 1WQC computation in \Cref{eq:description_1wqc_computation}. That is, the user must settle on a total ordering $>$ beforehand.\footnote{Compare to BOQC or 1WQC, where one can
define the total ordering $>$ during computation} Any two valid total orderings
will result in the same computation and have same requirement on the number of physical qubits,
but they might require different coherence time for the qubits, as we have illustrated on our work on the Grover algorithm~\cite{gustiani2019three}. 

\begin{algorithm}[h]
  \caption{Lazy 1WQC computation}
  \label{alg:lazy}
  \begin{algorithmic}[1]
      \Input{$\{(\mathcal G^>,I,O),f,\phi,\rho^{in}\}$}
  \Output{$\mathcal E(\rho^{in})$}\Comment{see \Cref{thm:lazy_correctness}}

  \Statex{\hspace{-2em}\itshape Conventions: 
      \begin{enumerate}[(I)]
\item Partial order $\succ$ is induced by flow $f$.
\item $z(i)\coloneqq\bigoplus_{k<i, i\in \neig{f(k)}} s_k$, $\invf{i}\equiv f^{-1}(i)$.
    \item $s_{\invf{i}}=0$ for all $i\in I$. 
    \end{enumerate}}
    \State{Assign $\rho^{in}$ to the input nodes $I$.}\label{ln:lazy_assign_i}
  \For{$i \in V$ with ordering $>$} \label{ln:lazy_loop}
  \For{$k \in A(i)$}\label{ln:lazy_assign_i} \Comment{see \Cref{eq:ai}}
      \State{assign state $\ket +$ to node $k$}
  \EndFor\label{ln:lazy_assign_e}
  \State{Apply entangling operations $E_{i\neig{i}}^>$.}
  \If{$i\in O^c$}
  \State{$\phi_i'\coloneqq (-1)^{s_{\invf{i}}}\phi_i+ z(i)\pi$} 
  \State{Measure $i$ in basis $\ket{\pm_{{\phi'}}}$ 
  and obtain measurement outcome $s_i$.}
  \Else
  \State{Correct output $i$ applying $X_i^{s_{\invf{i}}}Z_i^{z(i)}$.}
  \EndIf
  \EndFor\label{ln:lazy_endloop}
  \end{algorithmic}
\end{algorithm}

\Cref{alg:lazy} shows the general case, where the input and output are
quantum.\footnote{This is comparable to $\tilde I=I$ and $\tilde O=O$ in the
BOQC.}  If the input is classical, one can trivially encode as a quantum state $\rho^{in}$
implemented as the following. Given the input state as a bit string $c$,
one can implement by setting the input nodes as $\ket{+_{c_i\pi}}$ for all $i\in I$.
If the output is classical, all nodes in $O$ will be measured, \ie the
loop in line \ref{ln:lazy_loop} is replaced with 
\[\text{\textbf{for} $i\in V$ with ordering $>$ \textbf{do}},\]
and the algorithm is terminated after line
\ref{ln:lazy_endloop}.  Consider \Cref{exa:lazy} to illustrate running a lazy
1WQC computation.

\begin{example} \label{exa:lazy}
    Given a computation $\{(\mathcal G^>,I,O),f,\phi,\rho^{in}\}$, where
    $\mathcal G=(V,E)$ for $V=\{1,2,3,4,5,6,7\}$ and
    $E=\{(1,3),(2,3),(2,4),(4,6),(4,5),(3,5),(3,7)\}$; state $\rho^{in}$ is
    assigned to input nodes $I=\{1,2\}$, output nodes $O=\{5,6,7\}$, and
    quantum inputs-outputs are expected. Running this computation in the lazy
    1WQC scheme is illustrated in \Cref{fig:exa_lazy}.
    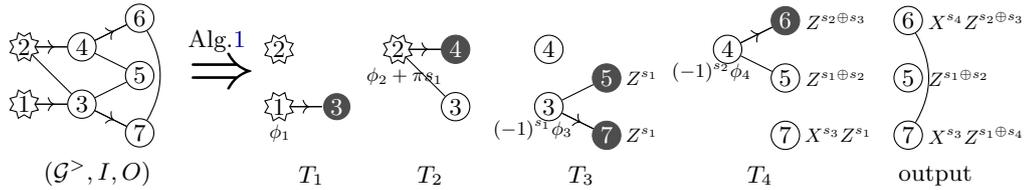
\begin{figure}[!h]
\resizebox{0.85\columnwidth}{!}{
\begin{tikzpicture}
 [font=\footnotesize,scale=.75,auto,every node/.style={circle,inner
  sep=.5pt,minimum size=5pt,draw=black,text width=7.0pt,align=center}]
\node[star,star points=8,star point ratio=0.7] (n1) at (0,0){1};
\node[star,star points=8,star point ratio=0.7] (n2) at (0,1){2};
\node[] (n3) at (1,1){4};
\node[] (n5) at (1,0){3};
\node[] (n4) at (2,1.5){6};
\node[] (n6) at (2,0.5){5};
\node[] (n7) at (2,-0.5){7};
\node[draw=none, label={above:Alg.\ref{alg:lazy}}] (na) at (3,0.5){\Huge$\Rightarrow$};
\draw[](n1) -- (n5);
\draw[](n2) -- (n3);
\draw[](n2) -- (n5);
\draw[](n3) -- (n4);
\draw[](n3) -- (n6);
\draw[](n4) to[out=-60,in=60] (n7);
\draw[](n5) -- (n6);
\draw[](n5) -- (n7);
\begin{scope}[decoration={markings,mark=at position 0.6 with {\arrow{>}}}]
\draw[postaction={decorate}](n1) -- (n5);
\draw[postaction={decorate}](n2) -- (n3);
\draw[postaction={decorate}](n3) -- (n4);
\draw[postaction={decorate}](n5) -- (n7);
\end{scope}
\node[draw=none] (ng) at (0.5,-1.2){$(\mathcal G^>,I,O)$};
\end{tikzpicture}
\hspace{1em}
\begin{tikzpicture}
 [font=\footnotesize,scale=.75,auto,every node/.style={circle,inner
  sep=.5pt,minimum size=5pt,draw=black,text width=7.0pt,align=center}]
  \node[label=below:\scalebox{0.8}{$\phi_1$}, star,star points=8,star point ratio=0.7] (n1) at (0,0){1};
\node[star,star points=8,star point ratio=0.7] (n2) at (0,1){2};
\node[draw=none] (n3) at (1,1){ };
\node[white, fill=black!70!white ] (n5) at (1,0){3};
\draw[](n1) -- (n5);
\begin{scope}[decoration={markings,mark=at position 0.6 with {\arrow{>}}}]
\draw[postaction={decorate}](n1) -- (n5);
\end{scope}
\node[draw=none] (ng) at (.5,-1.2){$T_1$};
\end{tikzpicture}
\begin{tikzpicture}
 [font=\footnotesize,scale=.75,auto,every node/.style={circle,inner
  sep=.5pt,minimum size=5pt,draw=black,text width=7.0pt,align=center}]
\node[draw=none] (n1) at (0,0){ };
\node[label={[yshift=-20pt,xshift=-8pt]\scalebox{0.8}{$\phi_2+\pi s_1$}},star,star points=8,star point ratio=0.7] (n2) at (0,1){2};
\node[white,fill=black!70!white, draw=black!70!white] (n3) at (1,1){4};
\node[] (n5) at (1,0){3};
\draw[](n2) -- (n3);
\draw[](n2) -- (n5);
\begin{scope}[decoration={markings,mark=at position 0.6 with {\arrow{>}}}]
\draw[postaction={decorate}](n2) -- (n3);
\end{scope}
\node[draw=none] (ng) at (.5,-1.2){$T_2$};
\end{tikzpicture}
\begin{tikzpicture}
 [font=\footnotesize,scale=.75,auto,every node/.style={circle,inner
  sep=.5pt,minimum size=5pt,draw=black,text width=7.0pt,align=center}]
\node[draw=none] (n1) at (0,0){ };
\node[draw=none] (n2) at (0,1){ };
\node[draw=none] (n4) at (2,1.5){ };
\node[] (n3) at (1,1){4};
\node[label={[xshift=-.6cm,yshift=-.7cm]\scalebox{0.8}{$(-1)^{s_1}\phi_3$}}] (n5) at (1,0){3};
\node[white,fill=black!70!white, draw=black!70!white,label=right:\scalebox{0.8}{$Z^{s_1}$}] (n6) at (2,0.5){5};
\node[white,fill=black!70!white, draw=black!70!white,label=right:\scalebox{0.8}{$Z^{s_1}$}] (n7) at (2,-0.5){7};
\draw[](n5) -- (n6);
\draw[](n5) -- (n7);
\begin{scope}[decoration={markings,mark=at position 0.6 with {\arrow{>}}}]
\draw[postaction={decorate}](n5) -- (n7);
\end{scope}
\node[draw=none] (ng) at (1.5,-1.2){$T_3$};
\end{tikzpicture}
\begin{tikzpicture}
 [font=\footnotesize,scale=.75,auto,every node/.style={circle,inner
  sep=.5pt,minimum size=5pt,draw=black,text width=7.0pt,align=center}]
\node[draw=none] (n1) at (0,0){};
\node[draw=none] (n2) at (0,1){};
\node[draw=none] (n5) at (1,0){};
\node[label={[xshift=-.6cm,yshift=-.7cm]\scalebox{0.8}{$(-1)^{s_2}\phi_4$}}] (n3) at (1,1){4};
\node[white,fill=black!70!white,draw=black!70!white,label=right:\scalebox{0.8}{$Z^{s_2\oplus s_3}$}] (n4) at (2,1.5){6};
\node[label=right:\scalebox{0.8}{$Z^{s_1\oplus s_2}$}] (n6) at (2,0.5){5};
\node[label=right:\scalebox{0.8}{$X^{s_3}Z^{s_1}$}] (n7) at (2,-0.5){7};
\draw[](n3) -- (n4);
\draw[](n3) -- (n6);
\begin{scope}[decoration={markings,mark=at position 0.6 with {\arrow{>}}}]
\draw[postaction={decorate}](n3) -- (n4);
\end{scope}
\node[draw=none] (ng) at (1.5,-1.2){$T_4$};
\end{tikzpicture}
\begin{tikzpicture}
 [font=\footnotesize,scale=.75,auto,every node/.style={circle,inner
  sep=.5pt,minimum size=5pt,draw=black,text width=7.0pt,align=center}]
  \node[draw=none] (n3) at (1,1){};
  \node[ label=right:\scalebox{0.8}{$X^{s_4}Z^{s_2\oplus s_3}$}   ] (n4) at (2,1.5){6};
\node[ label=right:\scalebox{0.8}{$Z^{s_1\oplus s_2}$}] (n6) at (2,0.5){5};
\node[ label=right:\scalebox{0.8}{$X^{s_3}Z^{s_1\oplus s_4}$}] (n7) at (2,-0.5){7};
\draw[](n4) to[out=-60,in=60] (n7);
\node[draw=none] (ng) at (2,-1.2){output};
\end{tikzpicture}
}
        \caption{
            Running \Cref{alg:lazy} with an open graph $(\mathcal G^>,I,O)$. The arrows indicate the flow of $\mathcal G$,
            the node number indicates total order $>$, $T_i$
            indicates the time step when measuring $i$, $s_i$ signifies
            the measurement outcome of measuring $i$, the grey nodes are fresh qubits
            initialized with state $\ket+$ before measuring $i$, and
            the Pauli correction is shown on the corresponding node.
            The highest number of physical qubits requirements is 4, occurring at time-steps $T_3$ and $T_4$.\label{fig:exa_lazy}}
    \end{figure}
\end{example}

The allocation of fresh physical qubits, which corresponds to grey nodes in
\Cref{exa:lazy}, occurs in \Cref{alg:lazy} in lines
\ref{ln:lazy_assign_i}--\ref{ln:lazy_assign_e}; these qubits are then
initialized with state $\ket +$. We denote such a set of nodes as 
\begin{equation}
    A(i)\coloneqq
    \neigc{i}\setminus(I\cup_{j<i}\neigc{j})\label{eq:ai},
\end{equation} 
which is a closed neighborhood, excluding the ones that have been assigned
before. We assume that the input nodes $i\in I$ are assigned with the desired
quantum input $\rho^{in}$ before the scheme starts.  As it is obvious that
$A(i)\subseteq V$, the lazy 1WQC scheme does not construct the whole graph
state at once.\footnote{There are certainly some cases where lazy 1WQC constructs
the whole graph, \eg graphs that are fully connected and star graphs.}

In the following, we establish the correctness of the lazy 1WQC.  That is, we
show that it results in the same computation as the standard 1WQC scheme. Then
we derive bounds on the number of physical qubits needed.

First, note that we can write the resulting pattern of \Cref{alg:lazy} by
consecutively placing the initialization, entanglement, Pauli-correction, and
measurement commands:
 \begin{equation}
   \mathcal P_{lazy}=
   \bigotimes_{j\in O}
   (X_j^{\invf{j}}Z_j^{z(j)})
   E_{O}\,
   \OP_{i\in O^c} M_i^{\phi_i}
   X_i^{\invf{i}}Z_i^{z(i)} 
   E^>_{i\neig{i}}N_{A(i)}^0,
   \label{eq:pattern_lazy}
 \end{equation}
 where $z(i)\coloneqq\bigoplus_{k<i, i\in\neig{f(k)}} s_k$ and  $\invf{i}=0$ for all
 $i\in I$. Note that the specification of Pauli operators before measurement
 follows~\Cref{eq:correction}.

 \Cref{thm:lazy_correctness} formally states the correctness of the lazy 1WQC,
 with the help of~\Cref{lem:mutually_disjoint,lem:partition,lem:correction}.
The proof of each lemma is available in \Cref{app:proofs}.

\begin{restatable}{lemma}{lemmutuallydisjoint}\label{lem:mutually_disjoint}
    \cite{gustiani}Suppose the open graph state $(\mathcal G^>,I,O)$ has flow $(f,\succ)$ and a proper total order $>$, 
    then $A(i)$ contains at least $f(i)$ for all $i\in O^c$ and $A(i)\cup A(j)=\varnothing$ for 
    all $i\neq j$.
\end{restatable}

\begin{restatable}{lemma}{lempartition}\label{lem:partition}
    \cite{gustiani}Suppose the open graph state $(\mathcal G^>,I,O)$ has flow $(f,\succ)$ and 
    a proper total order $>$, then $\cup_{{i\in V}}A(i)=I^c$.
\end{restatable}


\begin{restatable}{lemma}{leme}\label{lem:e}
    \cite{gustiani}Suppose an open graph state $(\mathcal G^>,I,O)$ has flow $(f,\succ)$ and a proper total order $>$, 
    then $\prod_{i\in V} E^>_{i\neig{i}} = E_{\mathcal G}$,
where $E^>_{i\neig{i}}\coloneqq \prod_{k\in \neig{i},k>i}E_{ik}$.
\end{restatable}

\Cref{lem:mutually_disjoint} shows that in every time-step, one must assign at
least one fresh qubit. Followed by \Cref{lem:partition} which shows that every
non-input qubit is assigned once. \Cref{lem:e} proves that we recover the whole
graph. Note that, these lemmas (\Cref{lem:mutually_disjoint,lem:partition})
consider the general case --- quantum input and output --- in which the input nodes are initialized
beforehand. 
Finally, we prove the correctness
of lazy 1WQC in the following theorem.

\begin{theorem} \label{thm:lazy_correctness}\cite{gustiani}
    The lazy 1WQC scheme and the 1WQC scheme implement the same map, they produce the same output for the same input.
\end{theorem}
\begin{proof}
    The following proof is similar to the one in~\cite{gustiani}. Let $\{(\mathcal G^>,I,O),f,\phi,\rho^{in}\}\coloneqq\mathscr I$ be the input of the lazy
    scheme (\Cref{alg:lazy}), where $>$ is consistent with the flow
    $(f,\succ)$, \ie the flow of the open graph $(\mathcal G,I,O)$. Since $>$
    is consistent with $\succ$, $\mathscr I$ is a valid input of the 1WQC scheme. 

    One can prove the map equality by comparing the patterns, \eg 
    we will show that $\mathfrak P_{1wqc}$ in~\Cref{eq:pattern} can be reduced to
    pattern $\mathfrak P_{lazy}$ in \Cref{eq:pattern_lazy}. Using previous results, 
    we have 
  \begin{align}\label{eq:3lines_boqco}
      \mathfrak P_{1wqc} 
     &=
    \POP_{i\in O^c}(X_{f(i)}^{s_i}
    \prod_{k\in\neig{f(i)}\setminus\{i\}}Z_i^{s_i} M_i^{\phi_i})E_{\mathcal G}N_{I^c}^0\\
    &\stackrel{\text{Lem.}~\ref{lem:correction}}{=}
    \prod_{j\in O}
    X_j^{s_{\invf{j}}}
    Z_j^{z(j)}
    \POP_{i\in O^c} ( 
    M_i^{\phi_i}
    X_i^{s_{\invf{i}}}
    Z_i^{z(i)}
    )E_{\mathcal G}N_{I^c}^0 
    \\
    &\stackrel{\text{Lems.}~\ref{lem:partition},~\ref{lem:e}}{=}
    \prod_{j\in O}
    X_j^{s_{\invf{j}}}
    Z_j^{z(j)}
    \OP_{i\in O^c} ( 
    M_i^{\phi_i}
    X_i^{s_{\invf{i}}}
    Z_i^{z(i)})
    E_O\,
    \OP_{k\in O^c} E_{k\neig{k}}^>
    \OP_{l\in O^c}N_{A(l)}^0 
    \\
    &=
    \bigotimes_{j\in O}
    X_j^{s_{\invf{j}}}
    Z_j^{z(j)}
    \OP_{i\in O^c} ( 
    M_i^{\phi_i}
    X_i^{s_{\invf{i}}}
    Z_i^{z(i)})
    E_O\,
    \OP_{k\in O^c} E_{k\neig{k}}^>
    \OP_{l\in O^c}N_{A(l)}^0.
\end{align}
Note that in the third equality, the partial ordering $\succ$ is replaced with
the total ordering $>$; this is valid since $>$ is consistent with $\succ$.  

Now we need to commute the entangling and preparation operators such that they
are distributed according to the lazy scheme. First, consider any two nodes $i$ and
$k$, where $i,k\in O^c$ and $i < k$. The preparation and entangling operators are commuting,
{\em i.e.},
\begin{equation}\label{eq:commute}
E^>_{k\neig{k}}E^>_{i\neig{i}}N_{A(k)}^0N_{A(i)}^0=
E^>_{k\neig{k}}N_{A(k)}^0E^>_{i\neig{i}}N_{A(i)}^0.
\end{equation}
However, we need to check if the condition of causality holds: there is no entanglement
operation involving qubits that are already measured or not yet
created.  Denote the set of edges $e(i)\coloneqq \{(i,k) \mid k\in\neig{i},
k>i\}$, \ie edges that correspond to entangling operations
$E^>_{i\neig{i}}$. By definition, set $e(i)$ does not contain any node
that has already been measured, namely any $k<i$. The nodes that correspond to edges $e(i)$ are
\begin{equation}\label{eq:n}
\{k\in\neigc{i}\mid k > i\}\eqqcolon n(i). 
\end{equation}
By definition, $A(i)$ contains all nodes in $\neigc{i}$ minus the ones that have already been created 
$I\cup_{k<i} \neigc{k}$, thus $\forall x\in e(i), x\in\{(i,j)\mid j\in \cup_{j\leq i}A(j) \}$,
which means every qubit connected by an edge in $e(i)$ is already initialized.  
Thus, there is no entanglement involving a qubit that has not yet been created. 
Therefore, \Cref{eq:commute} is causal.

Considering the measurement operator and the Pauli corrections, we need to 
commute the entangling operation through them, namely 
\begin{equation}
M_i^\phi X_i Z_i  E^>_{k\neig{k}}N_{A(k)}^0E^>_{i\neig{i}}N_{A(i)}^0=
E^>_{k\neig{k}}N_{A(k)}^0M_i^\phi X_i Z_i  E^>_{i\neig{i}}N_{A(i)}^0,
\end{equation}
which is true if and only if $i\not\in n(k)$. By definition of $n(k)$ (see \Cref{eq:n}),
$i < k$, thus, $i\not\in n(k)$. Thus, we can distribute the entangling and preparation 
operators in \Cref{eq:3lines_boqco} with respect to the ordering $>$ and obtain
\begin{equation}
    \bigotimes_{j\in O}
    (X_j^{s_{\invf{j}}}
    Z_j^{z(j)})
    E_O
    \OP_{i\in O^c} 
    M_i^{\phi_i}
    X_i^{s_{\invf{i}}}
    Z_i^{z(i)}
    E_{i\neig{i}}^>
    N_{A(i)}^0=\mathfrak P_{lazy}.
\end{equation}
\end{proof}

Since the minimal number of physical qubits is the pivot in the lazy 1WQC, it
is natural to ask for a bound on the number of physical qubits for an arbitrary 1WQC
computation.  We provide \Cref{conj:nqubit} to immediately suggest an answer.
The intuition behind \Cref{conj:nqubit} stems from a property of a graph with
flow; namely, the number of nodes per layer cannot shrink.  It is due to
non-colliding correction nodes: two distinct nodes $i,j$ can not have the same
$X$-correction node, $f(i)\neq f(j)$; otherwise, it violates the flow criteria
(\Cref{eqn:flow}).

\begin{conjecture}\label{conj:nqubit}
The number of physical qubits required to run lazy 1WQC in \Cref{alg:lazy},
regardless the input and output type --- whether classical or quantum ---
is bounded by $\abs{O} + 1$.   
\end{conjecture}

\subsection{The BOQCo protocol}

Executing a 1WQC computation within the lazy scheme reduces the number of
physical qubits vastly, bounded to $\abs{O} + 1$ per
\Cref{conj:nqubit}. Here we integrate the lazy 1WQC paradigm into BOQC,
producing a protocol that we call BOQCo (BOQC-optimized). BOQCo allows the
server to prepare the graph state as needed, employing the minimal number of
physical qubits while maintaining the blindness of the multi-party
scheme.\footnote{At this point, we claim that BOQC and BOQCo are
unconditionally blind; this statement is proven and exclusively discussed in
\Cref{sec:security}.}

The BOQCo protocol is shown in \Cref{pro:boqco}; in AC language, we address it
as $\pi_{boqco}=\{\conv{A},\conv{B},\conv{O}\}$.  The scheme employs a strategy identical to BOQC to provide blindness; it is apparent from the
introduced randomness: $r,t,\alpha,\beta$. The distinguishing feature of BOQCo
lies in the distribution of the computation, which follows the lazy 1WQC. 

\begin{protocol}[!h]
    \caption{BOQCo ($\pi_{boqco}=\{\conv{A},\conv{B},\conv{O}\}$)}
  \label{pro:boqco}
  \begin{algorithmic}[1]

      \Statex{\hspace{-2em}\bfseries Alice's input: $\{(\mathcal G, I, O),f,\phi,\rho^{in}_{\mathcal A}\}$}\Comment{$\tilde I=I$ and $\tilde O=O$}
  \Statex{\hspace{-2em}\bfseries Oscar's input: $\{\psi\}$}
  \Statex{\hspace{-2em}\bfseries Alice's output for an honest Bob:
  $\rho^{out}_{\mathcal A}=\mathcal E(\rho^{in}_{\mathcal A})$}

  \Statex{\itshape Assumptions and conventions: 
          \begin{enumerate}[(I)]
              \item Alice ($\mathcal A$) and Oscar ($\mathcal O$) have
                  performed pre-protocol steps in \Cref{def:prestep_boqc}; 
                  Bob knows $\{(\mathcal G,\tilde I,\tilde O),\vera,\vero,\succ,>,b\}$. 
                  Here, we set $\tilde I=I$ and $\tilde O=O$. Recall $\tilde O \cap \vero=\varnothing$
                  (quantum outputs are held by Alice) and $\Omega=\{\frac{\pi k}{2^{b-1}}\}_{0\leq k<2^b}$.

      \item $s_{\invf{i}}=0,\forall i\in I$ and $t_i=0,\forall i\in I^c$. 
          
      \item $\invf{i}\equiv f^{-1}(i)$,
          $z(i)\coloneqq\bigoplus_{k\prec i, i\in \neig{f(k)}}s_k$, and  $t(i)\coloneqq\bigoplus_{k\in I, i\in N_{\mathcal G}(k)}t_k$.
          \end{enumerate}}

\Statex{\hspace{-2em}\bfseries \circled{0} Pre-preparation}
  \State{Alice and Oscar receive keys $r,t$ via a secure key channel,
  where $r_i\in\{0,1\},~i\in O^c$ and $t_j\in\{0,1\},~j\in I$.}

\Statex{\hspace{-2em}\bfseries\centering \circled{1} BOQC by parts}
\For{$i \in V\setminus O$ with ordering $>$}

\For{$k\in A(i)\cup I $}\Comment{\Cref{eq:ai}, \Cref{sec:1wqc}}
        \If{$k\in I$}\Comment{input qubit}
\State{Alice applies $Z_k(\alpha_k)X^{t_k}_k$ to 
$\tr_{I\setminus k}[\rho^{in}_{\mathcal A}]$ and sends it to Bob, $\alpha_k\in\Omega$ is chosen at random.}

\State{Alice updates angles:
    \vspace{-1.5em}
    \begin{align*}
        \phi_k &=(-1)^{t_k}\phi_k \\
        \phi_j &=\phi_j+t_k \pi,~\forall j\in N_{\mathcal G}(k)\cap\vera.
\end{align*}
    \vspace{-1.5em}
}
            
    \State{Oscar updates angles: 
    $\psi_j=\psi_j+t_k \pi,~\forall j\in N_{\mathcal G}(k)\cap\vero$.}

    \ElsIf{$k\in O$} \Comment{output qubit}
            \State{Bob prepares $\ket{+}_k$.}

            \Else \Comment{auxiliary qubit}
            \If{$k\in \vera$}
                    \State{Alice prepares $\ket{+_{\alpha_k}}_k$, sends
                    it to Bob, $\alpha_k\in\Omega$ is chosen at random.}
    \algstore{boqco}
\end{algorithmic}
\end{protocol}
\begin{protocol}
  \begin{algorithmic}[1]
    \algrestore{boqco}

            \ElsIf{$k\in\vero$}

                    \State{Oscar prepares $\ket{+_{\beta_k}}_k$ and send it to Bob, $\beta_k\in\Omega$ is chosen at random.}
            \EndIf
        \EndIf
    \EndFor

    \State{Bob applies entangling operations $E^>_{i N_{\mathcal G}(i)}$.}
    \Comment{\Cref{eq:e_notation}, \Cref{sec:1wqc}}

    \If{$i\in O^c$}

    \If{$i\in \vera$}
        \State{Alice computes $\phi'_i=(-1)^{s_{\invf{i}}}\phi_i + z(i)\pi  $.}
        \State{Alice computes $\delta_i\coloneqq \phi'_i + \pi r_i + \alpha_i$ and sends Bob $\delta_i$.}
    \State{Bob measures $i$ in $\ket{\pm_{\delta_i}}$ basis, sends Alice and Oscar 
          the outcome $\tilde s_i$.}
    \State{Alice and Oscar set $s_i=\tilde s_i\oplus r_i$.}

    \ElsIf{$i\in\vero$}
        \State{Oscar computes $\psi'_i=(-1)^{s_{\invf{i}}}\psi_i + z(i)\pi$.}
        \State{Oscar computes $\delta_i\coloneqq \psi'_i + \pi r_i + \beta_i$ and sends Bob $\delta_i$.}
    \EndIf
    \Else\Comment{output qubit transmissions and corrections} 
    \State{Bob sends Alice qubit $i$.}
    \State{Alice corrects qubit $i$ by applying $X^{s_{\invf{i}}+t_i}_i  Z_i^{z(i)+t(i)}$.}
\EndIf
\EndFor

  \end{algorithmic}
\end{protocol}

In the BOQCo, we divide the process into three main steps: \circled{0}
pre-preparation that is identical to BOQC, \circled{1} computation part by
part, which is BOQC (excluding pre-preparation) done one part at a time, and
\circled{2} output transmission and correction that is also identical to BOQC.
The input and output of the BOQCo protocol is identical to the BOQC, \ie  
it receives input from Alice as $\{(\mathcal
G,I,O),f,\phi,\rho^{in}_{\mathcal A}\}$ and input from Oscar as $\psi$.

We define the correctness of a BOQCo computation as if it is run in the 1WQC
scheme. Formally, we state the correctness in \Cref{thm:correctness_boqco},
where the proof is provided in \Cref{app:proofs}.

\begin{restatable}{theorem}{thmcorrectnessboqco}\label{thm:correctness_boqco}\cite{gustiani}
    The BOQCo protocol $\pi_{boqco}$ defined in \Cref{pro:boqco} delegates a
    computation with the isometry defined in \cref{thm:pattern}
    for the same computation, without requiring Alice
    and Oscar to communicate their computation to each other.
\end{restatable}


In terms of quantum power between clients and servers, BOQCo has requirements identical
to those of BOQC as discussed in \Cref{sec:quantum_power}.  This is
because BOQC and BOQCo differ only in the ordering among qubit transmissions,
entanglements, and measurements.

\section{Security analysis}\label{sec:security}
This section elaborates on the security of BOQC and BOQCo using the AC
framework.  We promise composable \emph{blindness} for our protocols. For that,
we need to achieve two statements: correctness and security (blindness).  We
separately discuss each statement: \Cref{sec:correctness} for correctness and
\Cref{sec:malicious_bob} for blindness.  In the end, we investigate the
consequences of our security definition; the issue of leaked information in the ideal
resource puts limits on the permitted graph states of oracles.

\subsection{Correctness}\label{sec:correctness}

We use \Cref{def:correctness} to state the correctness --- also know as completeness --- of BOQC and BOQCo
protocols.  Here, we prove that both protocols are perfectly correct
($\varepsilon=0$) and realize an ideal resource $\sboqc$, denoted as $\euscr
R\xrightarrow{\pi_{boqc},0}\sboqc$ and $\euscr
R\xrightarrow{\pi_{boqco},0}\sboqc$, where $\euscr R$ is the real-world
resource connected to the protocols and $\sboqc$ is defined in
\Cref{fig:sboqc}.  In both protocols, resource $\euscr R$ comprises a secure
key channel, quantum channels, and insecure classical channels.

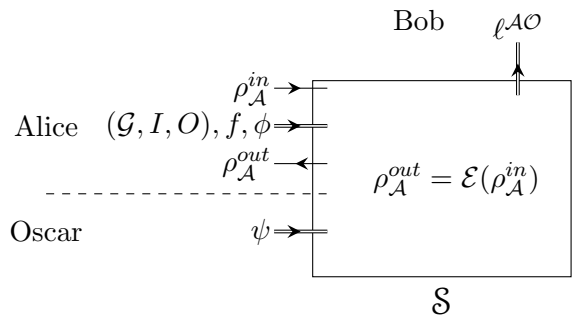
\begin{figure}[htbp]
  \centering
\begin{tikzpicture}
    \node[draw, rectangle, minimum width = 3.4 cm, minimum height = 2.6 cm] (box) at (2,0) 
  {
  \hspace{1em}$\rho^{out}_{\mathcal A}=
       \mathcal E(\rho^{in}_{\mathcal A})$};
       \node[below] at (box.south) {\scaleeq[1.3]{\sboqc}};
\begin{scope}[auto,every node/.style={draw=none,
  inner sep=1em,align=center}, ->-/.style={decoration={markings,mark=at position #1 with {\arrow[scale=1.5,>=stealth]{>}}},postaction={decorate}}]
  \draw[->-=.5] ([yshift=1.2cm,xshift=-.5cm]box.west) -- node[left]{$\rho^{in}_{\mathcal A}$}([yshift=1.2cm,xshift=0.2cm]box.west);
  \draw[->-=.6] ([yshift=0.2cm,xshift=.2cm]box.west) --node[left]{$\rho_{\mathcal A}^{out}$}([yshift=0.2cm,xshift=-.5cm]box.west);
\end{scope}

\begin{scope}[auto,every node/.style={draw=none, inner sep=1em,align=center},
  ->-/.style={decoration={markings,mark=at position #1 with {\arrow[scale=0.9,>=stealth]{>}}},postaction={decorate}}]
  \draw[->-=.5,double] ([yshift=0.7cm,xshift=-.5cm]box.west)--node[left]{$(\mathcal G,I,O),f,\phi$}([yshift=0.7cm,xshift=0.2cm]box.west);
  \draw[->-=.5,double] ([yshift=-.7cm,xshift=-.5cm]box.west)--node[left]{$\psi$}([yshift=-.7cm,xshift=0.2cm]box.west);
  \draw[->-=.6,double] ([yshift=-.2cm,xshift=1cm]box.north)--node[above]{$\ell^{\mathcal{AO}}$}([yshift=.5cm,xshift=1cm]box.north);
\end{scope}

\draw[-,dashed] ([yshift=-0.2cm,xshift=-3.5cm]box.west)--node[left]{}([yshift=-0.2cm,xshift=0.2cm]box.west);
\node[draw=none] at ([yshift=0.7cm,xshift=-3.5cm]box.west){Alice};
\node[draw=none] at ([yshift=-0.7cm,xshift=-3.5cm]box.west){Oscar};
\node[draw=none] at ([yshift=.8cm, xshift=-.3cm]box.north){Bob};
\end{tikzpicture}

\caption{\cite{gustiani}The ideal BOQC resource $\sboqc$ in the absence of an
    adversary. The left side is the interface of clients, and the top
    side is the server interface.  Single- and double- line arrows indicate
    quantum and classical information, respectively.  The resource receives
    inputs from Alice: an open graph $(\mathcal G,I,O)$ with flow $f$, a quantum input
    $\rho^{in}_{\mathcal A}$, and a set of measurement angles $\phi_i\in\Omega$,
    where $\Omega=\{\frac{k\pi}{2^{b-1}}\}_{0\leq k<2^b}$ for an integer $b$.
    Also, it receives input $\psi_i\in\Omega$ from Oscar, and it does not take
    any inputs from Bob.
    Alice receives the final
    output $\rho_{\mathcal A}^{out}=\mathcal E(\rho^{in}_{\mathcal A})$, where
    $\mathcal E$ is the resulting superoperator of the algorithm, \ie $\mathcal
    E$ is the isometry described in \Cref{thm:pattern}. 
    Some classical information
    $\ell^{\mathcal{AO}}=\{(\mathcal G, \tilde I, \tilde O), f, \vera, \vero,\succ,>,b\}$ leaks on
Bob's interface, which corresponds to
public information in \Cref{def:prestep_boqc}, and is necessary to set up the protocol.}

  \label{fig:sboqc}
\end{figure}

The ideal resource $\sboqc$ describes the BOQC system in the ideal world when
Bob is honest. Resource $\sboqc$ has an identical description of inputs and
outputs with the BOQC protocol in \Cref{pro:boqc}. Resource $\sboqc$ also
describes the BOQCo system in the ideal world. It also has an identical
configuration of inputs and outputs with \Cref{pro:boqco}.  The leaked
information $\ell^{\mathcal{AO}}$ is not apparent in the protocols
(\Cref{pro:boqc} and \Cref{pro:boqco}); however this leakage is revealed in the proofs of the correctness
theorems: \Cref{thm:boqc_correctness} for the BOQC and
\Cref{thm:boqco_correctness} for the BOQCo.


\begin{theorem}
    \label{thm:boqc_correctness}\cite{gustiani}
    The BOQC protocol $\pi_{boqc}=(\conv{A},\conv{O},\conv{B})$ defined in 
    \Cref{pro:boqc} is perfectly correct and emulates the ideal resource $\sboqc$
    defined in \Cref{fig:sboqc}.
\end{theorem}
\begin{proof}
The following proof is similar to the one in~\cite{gustiani}. Protocol
$\pi_{boqc}$ is correct if it satisfies \Cref{def:correctness}:
$d(\conv{A}\conv{O}\euscr R\conv{B},\sboqc)=0$, \ie resources
$\conv{A}\conv{O}\euscr R\conv{B}$ and $\sboqc$ must be perfectly
indistinguishable, where $d$ is a pseudo-metric with properties discussed in
\Cref{sec:ac}.  Which means, we show that the distinguishing
advantage (defined in \Cref{eq:distinguishing_advantage}) is zero. For that, we
show that both resources have the same --- or statistically the same --- inputs
and outputs.  

First, we show that $\pi_{boqc}$ and $\sboqc$ have an identical description of inputs.
As shown in \Cref{pro:boqc}, the protocol receives inputs $\{(\mathcal G, I,
O),f,\phi,\rho_{\mathcal A}^{in}\}$ from Alice, $\psi$ from Oscar, and no honest
input from Bob. These inputs are identical to the inputs of $\sboqc$.

Second, $\sboqc$ sends an output $\rho^{out}_{\mathcal A}=\mathcal
E(\rho_{\mathcal A}^{in})$, where $\mathcal E$ is the superoperator with an
isometry given in \Cref{thm:pattern}; using \Cref{thm:correctness_boqc}, BOQC
also implements that isometry.  

Finally we show that BOQC leaks the same information as $\sboqc$, which is
$\ell^{\mathcal{AO}}$. Bob receives information $\{(\mathcal G, \tilde I, \tilde O), \vera,
\vero,\succ,>,b\}$ in the protocol, which is public information obtained
from the pre-protocol steps defined in \Cref{def:prestep_boqc} before the
protocol starts. The public information is identical to leak
$\ell^\mathcal{AO}$.  Bob is not curious in this setting, thus there is no
additional information obtained beyond $\ell^\mathcal{AO}$. 
\end{proof}
We note here that Bob
does not erase the received information, as it is required to run the protocol
and to provide the bill for his clients.\footnote{Bob is not curious, but he
needs to record some information for his clients to pay. For example, the dense
$\Omega$ may cost more than the sparse $\Omega$.}

The resource $\euscr S$ models a general case in which Alice expects quantum
input and quantum output, \ie $\tilde I=I$ and $\tilde O=O$. If Alice needs
only classical outputs, Bob will measure all output nodes $O$ and her output
density matrix $\rho_{\mathcal A}^{out}$ has diagonal form in the security
model $\euscr S$.  The same applies to the classical input, \eg for a bit
string $c=c_n\dots,c_2c_1$, we set $\rho^{in}_{\mathcal
A}=\bigotimes_{i=1}^{n}\ketbra{c_i}$ in the security model.  Note that, per
\Cref{def:prestep_boqc}, Bob knows beforehand the input-output type, which is
captured in the leaked information $\tilde I$ and $\tilde O$. For instance, Bob
knows the input is entirely classical if $\tilde I=\varnothing$. Such
information is necessary for Bob to prepare his channel.

As the counterpart of \Cref{thm:boqc_correctness}, we prove the correctness of
BOQCo protocol in \Cref{thm:boqco_correctness}:

\begin{theorem}
    \label{thm:boqco_correctness}\cite{gustiani}
    The BOQCo protocol $\pi_{boqco}=(\conv{A},\conv{O},\conv{B})$, defined in 
    \Cref{pro:boqco} is perfectly correct, and emulates ideal resource $\sboqc$
    defined in \Cref{fig:sboqc}.
\end{theorem}
\begin{proof}
    The following proof is similar to the one in~\cite{gustiani}.
    Using \Cref{def:correctness}, correctness is
    achieved when $d(\conv{A}\conv{O}\euscr R\conv{B},\sboqc)=0$. We prove this
    condition by reducing BOQCo to BOQC.

    First, \Cref{pro:boqco} (BOQCo) has the same inputs as \Cref{pro:boqc}
    (BOQC). Applying \Cref{thm:correctness_boqco}, BOQCo also results in the same
    computation as BOQC, thus, the same output.  Secondly, BOQCo and BOQC differ only in the ordering
    among transmissions, entanglements, and measurements; thus, there is no
    additional leak introduced beyond $\ell^{\mathcal{AO}}$ (\Cref{fig:sboqc}). 
    In terms of correctness, BOQCo is
    reducible to BOQC. Finally, since BOQC is perfectly correct within the
    composable definitions, BOQCo is also perfectly correct within the composable definitions.
\end{proof}

\subsection{Blindness: when Bob is malicious}\label{sec:malicious_bob}
\begin{figure}[!h]
  \centering
\begin{tikzpicture}
    \node[draw, rectangle, minimum width = 3.4 cm, minimum height = 2.6 cm] (box) at (2,0) 
  {
  \hspace{1em}$\rho^{out}_{\mathcal A}=
       \tilde{\mathcal E}(\tilde\rho^{in}_{\mathcal{AB}})$};
       \node[below] at (box.south) {\scaleeq[1.3]{\sboqcp}};
\begin{scope}[auto,every node/.style={draw=none,
  inner sep=1em,align=center}, ->-/.style={decoration={markings,mark=at position #1 with {\arrow[scale=1.5,>=stealth]{>}}},postaction={decorate}}]
  \draw[->-=.5] ([yshift=1.2cm,xshift=-.5cm]box.west) -- node[left]{$\rho_{\mathcal A}^{in}$}([yshift=1.2cm,xshift=0.2cm]box.west);
  \draw[->-=.6] ([yshift=.5cm,xshift=-1cm]box.north)--node[above]{$\squared{B}$}([yshift=-.2cm,xshift=-1cm]box.north);
  \draw[->-=.6] ([yshift=0.2cm,xshift=.2cm]box.west) --node[left]{$\rho_{\mathcal A}^{out}$}([yshift=0.2cm,xshift=-.5cm]box.west);
\end{scope}

\begin{scope}[auto,every node/.style={draw=none, inner sep=1em,align=center},
  ->-/.style={decoration={markings,mark=at position #1 with {\arrow[scale=0.9,>=stealth]{>}}},postaction={decorate}}]
  \draw[->-=.5,double] ([yshift=0.7cm,xshift=-.5cm]box.west)--node[left]{$(\mathcal G,I,O),f,\phi$}([yshift=0.7cm,xshift=0.2cm]box.west);
  \draw[->-=.5,double] ([yshift=-0.7cm,xshift=-.5cm]box.west)--node[left]{$\psi$}([yshift=-0.7cm,xshift=0.2cm]box.west);
  \draw[->-=.6,double] ([yshift=-.2cm,xshift=1cm]box.north)--node[above]{$\ell^{\mathcal{AO}}$}([yshift=.5cm,xshift=1cm]box.north);
\end{scope}

\draw[-,dashed] ([yshift=-0.2cm,xshift=-3.5cm]box.west)--node[left]{}([yshift=-0.2cm,xshift=0.2cm]box.west);
\node[draw=none] at ([yshift=0.7cm,xshift=-3.5cm]box.west){Alice};
\node[draw=none] at ([yshift=-0.7cm,xshift=-3.5cm]box.west){Oscar};
\node[draw=none] at ([yshift=1.1cm]box.north){Bob};

\end{tikzpicture}

\caption{\cite{gustiani} Ideal resource $\sboqcp$, the ideal-world resource
    when Bob is malicious.  The input and output setting on the clients'
    interfaces is identical to $\sboqc$ --- \Cref{fig:sboqc}. Resource
    $\sboqcp$ does not take any honest input from Bob, but he may entangle
    Alice's input with his state ($\tilde\rho^{in}_{\mathcal{AB}}$) and apply a
    superoperator of his choice ($\mathcal E$), obtaining $\tilde{\mathcal
    E}(\tilde\rho^{in}_{\mathcal{AB}})$.  Variable $\squared{B}$ captures 
    all Bob's dishonest inputs, which determines the final output $\rho_{\mathcal A}^{out}$.
    Input $\squared{B}$ is used to send the extension $\rho^{in}_{\mathcal
    B}$ to $\sboqcp$, and to determine map $\tilde{\mathcal E}$.  } \label{fig:sboqcp}
\end{figure}
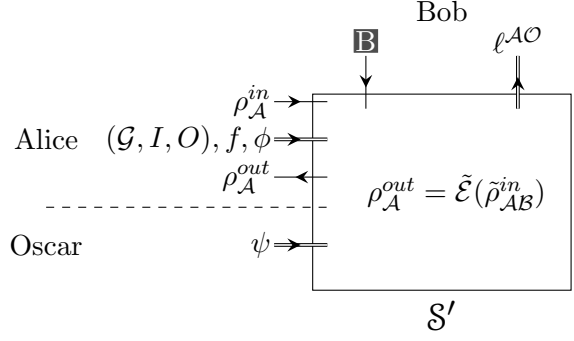

It remains to provide the security --- also known as soundness --- statement for
BOQC and BOQCo protocols to achieve a composable secure definition. The
security that is aimed for is \emph{perfect blindness}, meaning the adversary (Bob) can
learn nothing about the computation or the measurement outcomes.  The same
principles used for proving correctness apply also to prove blindness. We set the
security model in the ideal world that captures the desired blindness, and then
prove that our protocols that live in the real world are indistinguishable to
the ideal-world model. However, while the correctness captures the system when
everyone is honest, the blindness captures the situation when in the presence of
an adversary, \ie when Bob cheats.

In the presence of an adversary, our protocols realize the ideal resource
$\sboqcp$ that is defined in \Cref{fig:sboqcp}.  Resource $\sboqcp$ models the
ideal system in the ideal world when Bob is malicious.  On the clients' side,
$\sboqcp$ has the same input and output configuration as $\sboqc$; however,
in resource $\sboqcp$, Bob provides dishonest inputs as he wishes, which is 
captured in $\squared{B}$. Nevertheless, both resources $\sboqc$ and $\sboqcp$
leak the same information $\ell^{\mathcal{AO}}$.

Given that $\euscr R$ is the real-world resource used in our protocols, we need
to achieve statements $\euscr R\xrightarrow{\pi_{boqc},0}\sboqcp$ and $\euscr
R\xrightarrow{\pi_{boqco},0}\sboqcp$ where $\varepsilon=0$ signifies perfect
blindness.  To prove that, we must satisfy \Cref{def:security}, \ie there
exists a simulator $\simb$, such that $d(\conv{A}\conv{O}\euscr
R,\sboqcp\simb)=0$. Recall that a simulator $\simb$ is needed to make $\sboqcp$
and $\conv{A}\conv{O}\euscr R$ become comparable, \ie $\conv{A}\conv{O}\euscr
R$ has more inputs and outputs than $\sboqcp$. See the proof of
\Cref{thm:boqc_blindness} for explicit details.  \Cref{thm:boqc_blindness}
provides the security statement of the BOQC protocol, whose relaxation
$\sboqcp\simb$ is defined in~\Cref{pro:simulator} and in \Cref{app:protocols}.

\begin{theorem}\label{thm:boqc_blindness}\cite{gustiani}
    The BOQC protocol with dishonest Bob $\pi_{boqc}=\{\conv{A},\conv{O}\}$,
    defined in \Cref{pro:boqc}, is perfectly blind and realizes the ideal resource
    $\sboqcp$ defined in \Cref{fig:sboqcp}.
\end{theorem}
\begin{proof}
    The following proof is similar to the one in~\cite{gustiani}.
    Applying \Cref{def:security}, $\euscr R\xrightarrow{\pi_{boqc},0}\sboqcp$ if
    there exists a simulator $\simb$ such that $d(\conv{A}\conv{O}\euscr
    R,\sboqcp\simb)=0$; thus, we must find a relaxation $\sboqcp\simb$ that is
    perfectly indistinguishable from $\conv{A}\conv{O}\euscr R$.  Suppose
    the relaxation $\sboqcp\simb$ is defined in \Cref{pro:simulator}, then we proceed to prove
    that it is indistinguishable from $\conv{A}\conv{O}\euscr R$.

    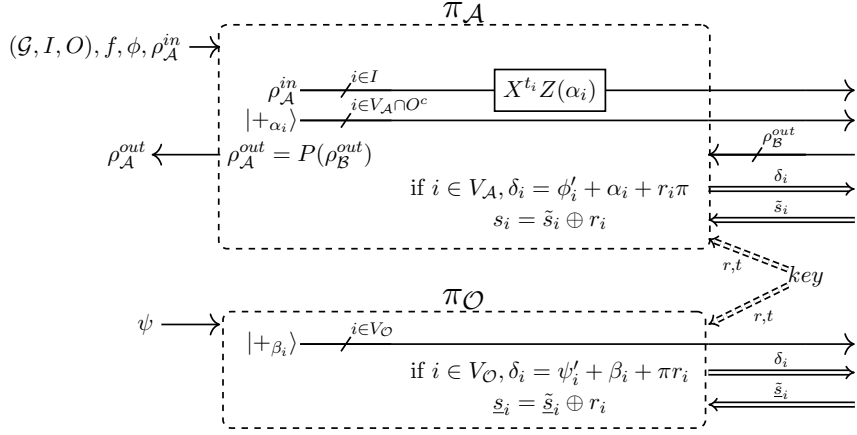
\begin{figure}[!h]
        \centering
        \scalebox{.8}{%
\begin{quantikz}[row sep=0.4em, inner sep=0em, column sep=0.8em]
    {(\mathcal G,I,O),f,\phi,\rho^{in}_{\mathcal A}~}\arrow[rightarrow,r,thick]&[2mm]\hphantom{}\gategroup[wires=7,steps=6,style={dashed,rounded corners, inner xsep=-0pt, inner ysep=4pt}]{\scaleeq[1.5]{\conv{A}}}&[-5mm]&&&&&[1cm]&[2mm] \\
&&& \lstick{$\rho^{in}_{\mathcal A}$}&\qw\qwbundle {i\in I} &\gate{X^{t_i} Z(\alpha_i)}&\qw&\qw\arrow[r,thick]&\\
&&& \lstick{$\ket{+_{\alpha_i}}$}&\qw\qwbundle {i\in \vera\cap O^c} &\qw&\qw&\qw\arrow[r,thick]&\\[2mm]
\phantom{rando}\rho^{out}_{\mathcal A}\:&\arrow[rightarrow,l,thick]&&\rho^{out}_{\mathcal A}= P(\rho^{out}_{\mathcal B})&&&&\arrow[l,thick]\qw\qwbundle {\rho^{out}_{\mathcal B}}&\qw\\ 
&&&&&\text{if}~i\in\vera,\delta_i=\phi_i'+\alpha_i+r_i\pi &\arrow[Rightarrow,rr,thick]{}{\delta_i}&&&&\\
&&&&&s_i=\tilde s_i\oplus r_i&&&\arrow[Rightarrow,ll,thick,swap]{}{\tilde s_i}&&\\
&&&&&&&&&&\\[3mm] 
&&&&&&&{key}\arrow[Rightarrow,ul,dashed,thick]{}{r,t}\arrow[Rightarrow,dl,dashed,thick]{r,t}{r,t}&&&\\[3mm] 
\phantom{randomst}{\psi}~\arrow[r,thick]&\hphantom{}\gategroup[wires=5,steps=6,style={dashed,rounded corners, inner xsep=-2pt, inner ysep=1pt}]{\scaleeq[1.5]{\conv{O}}}&&&&&&\\
    &&& \lstick{$\ket{+_{\beta_i}}$}&\qw\qwbundle {i\in \vero} & \qw &\qw&\qw\arrow[r,thick]&\\[.3em]
&&&&&\text{if}~i\in\vero,\delta_i=\psi'_i+\beta_i+\pi r_i&\arrow[Rightarrow,rr,thick]{}{\delta_i}&&&&\\
&&&&&\underline s_i=\underline{\tilde s}_i\oplus r_i&&&\arrow[Rightarrow,ll,thick,swap]{}{\underline{\tilde s}_i}&&\\ 
&&&&&&&&&&
\end{quantikz}
    }
    \caption[Resource $\conv{A}\conv{O}\res$]{ Pictorial representation of
        $\conv{A}\conv{O}\res$.  Each variable corresponds to the one in
        \Cref{pro:boqc}.  We distinguish measurement outcomes on Alice's interface
        ($\tilde s_i$) and on Oscar's interface ($\underline{\tilde s}_i$), because Bob is allowed to be
        dishonest, thus $\tilde s_i$ and $\underline{\tilde s}_i$ may be different. The left-hand
        side denotes the outside interface of the converters. The right-hand
        side indicates the inside interface connected to resources: secure key
        channel (dashed double arrow), classical channel (single arrow), and
        quantum channel (double arrow). The strike-through line indicates
        multiple qubits.} 
        \label{fig:convs1}
    \end{figure}

    To simplify the problem, we first reduce the protocol to a one-client
    protocol as follows. Consider a pictorial representation of $\pi_{boqc}$ (in the
    absence of Bob) in \Cref{fig:convs1}; it clearly shows that the common
    information between Alice and Oscar is the keys $s,t$, shared via a secure
    key channel.  Since it is known that a secure key channel guarantees secrecy and
    authenticity, Alice and Oscar obtain their keys without leaking any
    information to Bob. Therefore, we may think
    that Alice and Oscar have already shared the keys before the protocol
    starts. Thus, protocol $\conv{A}\conv{O}$ is reducible to a one-client
    protocol $\conv{AO}$, shown in \Cref{fig:convs2}, where the alternating
    part between Alice and Oscar is captured within functions
    $\theta(i)$, $s(i)$, and $\delta(i)$.

    Relaxation $\sboqcp\simb$ defined in \Cref{pro:simulator} is pictured in
    \Cref{fig:simulator}.  Thus, we now can prove the statement
    $d(\conv{A}\conv{O}\euscr R,\sboqcp\simb)=0$ by showing that
    \Cref{fig:convs2} and \Cref{fig:simulator} are indistinguishable.

    Notice that \Cref{fig:convs1,fig:convs2} clearly show identical inputs and outputs,
    indicated by the same configuration of arrows. 
    Thus, it now remains for us to prove that the arrows with
    a circled letter are (statistically) the same.

    Remark that we focus only on blindness without verifiability. Since
    verification is not involved, Alice does not care whether her computation
    is correct or not.  That is, some information related to the computation can be
    arbitrary, such as $\rho_{\mathcal B}^{out}$, $\tilde s_i$,
    $\tilde{\underline s}_i$, and $\delta_i$.

    Consider the information sent by Bob: $\circledl{c}~\rho^{out}_{\mathcal
    B}$, $\circledl{e}~\tilde s_i$, and $\circledl{f}~\underline{\tilde s}_i$. First, in both
    figures, $\rho^{out}_{\mathcal B}$ is an arbitrary state chosen by Bob (simulator).
    Second, in \Cref{fig:convs2}, $\tilde s_i$ and $\underline{\tilde s}_i$ signify measurement
    outcomes seen by Bob, which is random information independent of the actual
    measurement outcomes: $s_i=\tilde s_i\oplus r_i$ and $\underline s_i=\underline{\tilde s}_i\oplus r_i$.  The same
    is true in \Cref{fig:simulator}; $\tilde s_i,\tilde{\underline s}_i$ are arbitrary information
    inputted to Bob's interface in $\sboqcp$.

    We now analyze the information received by Bob: $\circledl{a}$,
    $\circledl{b}$, and $\circledl{d}~\delta_i$.  First, in both figures,
    $\delta_i$ are uniformly distributed random angles, independent of the
    actual measurement angles, $\phi'_i,\psi'_i$.  Second, consider the
    information at $\circledl{a}$, namely at input node $i\in I$.  In
    \Cref{fig:convs2}, $\circledl{a}$ is an encrypted input state
    $X^{t_i}Z_i(\alpha_i)(\rho_i)$, where $\rho_i\coloneqq \tr_{I\setminus
    i}[\rho_{\mathcal A}^{in}]$.  In \Cref{fig:simulator}, $\circledl{a}$ is an
    uncorrected teleported state
    $X^{t_i}_iZ_i^{r_i}Z_i(\delta_i-\phi'_i)(\rho_i)=
    X^{t_i}_iZ_i(\delta_i-\phi'_i-\pi r_i)(\rho_i)=
    X^{t_i}_iZ_i(\alpha_i)(\rho_i)$, which is identical to $\circledl{a}$ in
    \Cref{fig:convs2}. Third, for $\circledl{b}$, consider $i\in \vera$. In
    \Cref{fig:convs2}, $\circledl{b}$ is $\ket{+_{\alpha_i}}$. In
    \Cref{fig:simulator}, $\circledl{b}$ is a remote state
    preparation\cite{bennett2001remote} of $Z^{r_i}\ket*{+_{\delta_i-\phi'_i}}=
    \ket*{+_{\delta_i-\phi'_i-\pi r_i}}= \ket*{+_{\alpha_i}}$. This is also
    clearly true for $i\in\vero$.

    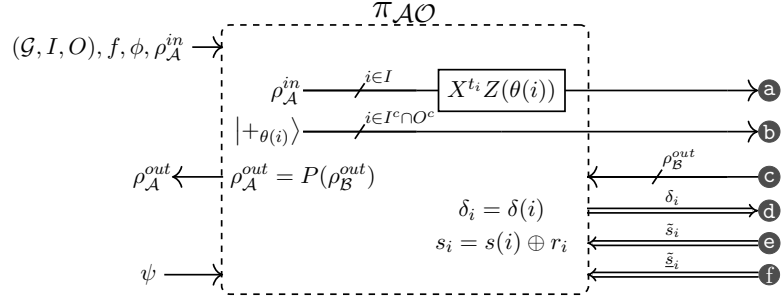
\begin{figure}[!h]
        \centering
        \scalebox{.8}{%
\begin{quantikz}[row sep=0.4em, inner sep=0em, column sep=0.8em]
    {(\mathcal G,I,O),f,\phi,\rho^{in}_{\mathcal A}~}\arrow[rightarrow,r,thick]&[2mm]\hphantom{}\gategroup[wires=7,steps=6,style={dashed,rounded corners, inner xsep=0pt, inner ysep=4pt}]{\scaleeq[1.5]{\conv{AO}}}&[-5mm]&&[3.6mm]&&&[2cm]&[2mm] \\
    &&& \lstick{$\rho^{in}_{\mathcal A}$}&\qw\qwbundle {i\in I} &\gate{X^{t_i} Z(\theta(i))}&\qw&\qw\arrow[r,thick]&\circledl{a}\\
    &&& \lstick{$\ket{+_{\theta(i)}}$}&\qw\qwbundle{i\in I^c\cap O^c} &\qw&\qw&\qw\arrow[r,thick]& \circledl{b}\\[2mm]
    \phantom{randomte}\rho^{out}_{\mathcal A}&\arrow[rightarrow,l,thick]&&\rho^{out}_{\mathcal A}=P(\rho^{out}_{\mathcal B})&&&&\arrow[l,thick]\qw\qwbundle {\rho^{out}_{\mathcal B}}&\qw\circledl{c}\\ 
    &&&&&  \delta_i=\delta(i) &\arrow[Rightarrow,rr,thick]{}{\delta_i}&&\circledl{d}\\ 
    &&&&& s_i=s(i)\oplus r_i &&&\arrow[Rightarrow,ll,thick,swap]{}{\tilde s_i}\circledl{e}\\
    \phantom{randomte}{\psi~}\arrow[rightarrow,r,thick]&\hphantom{}&&&&&&&\arrow[Rightarrow,ll,thick,swap]{}{\underline{\tilde s}_i}\circled{f} 
\end{quantikz}
    }
    \caption[Resource $\conv{AO}\res$]{Pictorial representation of $\conv{AO}\res$. 
    See \Cref{fig:convs1} for notation; additionally,\\ 
    $\theta(i)=
        \begin{cases}
            \alpha_i,\text{if}~i\in\vera \\
            \beta_i,\text{if}~i\in\vero, 
       \end{cases}$
    $s(i)=
        \begin{cases}
            \tilde s_i,\text{if}~i\in\vera \\
        \underline{\tilde s}_i,\text{if}~i\in\vero, 
       \end{cases}$ and~ 
    $\delta(i)=\begin{cases}
        \phi'_i + \alpha_i + \pi r_i,\text{if}~i\in\vera \\
        \psi'_i + \beta_i  + \pi r_i,\text{if}~i\in\vero. 
    \end{cases}$  
    }
    \label{fig:convs2}
    \end{figure}
    \begin{figure}[!h]
        \centering
        \scalebox{.8}{%
\begin{quantikz}[row sep=0.3em, inner sep=0em, column sep=0.8em]
    \hphantom{X}\gategroup[wires=7,steps=13, style={dashed,rounded corners,inner xsep=0pt,inner ysep=4pt}]{\scaleeq[1.5]{\simb}}&\lstick{$\ket0$}&\qw\qwbundle{i\in I}&\gate{H}&\ctrl{1}&\qw&\qw&\qw&\qw&\qw&\qw&\qw&\qw&\qw\arrow[rightarrow,r,thick]&\circledl{a}\\
&\lstick{$\ket1$}&\qw\qwbundle{i\in I}&\qw&\targ{}&\qw& &\lstick{$\ket 0$}&\qw\qwbundle{i\in I^c\cap O^c}&\gate{H}&\ctrl{1}&\qw&\qw&\qw\arrow[rightarrow,r,thick]&\circledl{b}\\
&&&&&\arrow[dash,u,thick]&&\lstick{$\ket 1$}&\qw\qwbundle{i\in I^c\cap O^c}&\qw&\targ{1}&\qw&&&&\\
&&&&&\arrow[dash,u,thick]&&&&&&\arrow[dash,u,thick]&&\qw\qwbundle{}&\arrow[rightarrow,ll,thick,swap]{}{\rho^{out}_{\mathcal B}}\circledl{c}\\
&&&&&\arrow[dash,u,thick]&&&\text{\small pick}~\delta_i\in\Omega&&&\arrow[dash,u,thick]&\arrow[Rightarrow,rr,thick]{}{\delta_i}&&\circledl{d}\\
&&&&&\arrow[dash,u,thick]&&&\text{\small at random}&&&\arrow[dash,u,thick]&&&\arrow[Rightarrow,ll,thick,swap]{}{\tilde s_i}\circledl{e}\\
&&&&&\arrow[dash,u,thick]&&&&&&\arrow[dash,u,thick]&&&\arrow[Rightarrow,ll,thick,swap]{}{\tilde{\underline s}_i}\circledl{f}\\
&&&&&\arrow[rightarrow,d,thick]\arrow[dash,u,thick]&&\arrow[rightarrow,d,thick]{}{\scaleeq[1.5]{\rho^{out}_{\mathcal B}}}&\arrow[Rightarrow,d,thick]{}{\scaleeq[1.5]{\delta_i,\tilde s_i,\tilde{\underline s}_i}}&&&\arrow[rightarrow,d,thick]\arrow[dash,u,thick]&&&\\[1cm]
&&&&&\arrow[dash,u,thick]&&&&\arrow[Rightarrow,u,thick,swap]{}{\scaleeq[1.5]{\ell^{\footnotesize\mathcal AO}}}  &&\arrow[dash,u,thick]&&&  \\[-1mm]
&&&(\mathcal G,I,O),f,\phi,\rho^{in}_{\mathcal A}\arrow[rightarrow,r,thick]&\hphantom{}\gategroup[wires=3,steps=12,style={dashed,rounded corners,inner xsep=-4pt, inner ysep=4mm}, label style={label position=below,anchor=north,yshift=-0.25cm}]{\scaleeq[1.5]{\sboqcp}}&\arrow[dash,u,thick]&\qw\qwbundle{}&\targ{}&\qw&\meter{$t_i$}&&\arrow[dash,u,thick]&&&&  \\
&&&\phantom{random}{\rho^{out}_{\mathcal A}}~&\arrow[rightarrow,l,thick]&\rho^{in}_{\mathcal A}&\qw\qwbundle{}&\ctrl{-1}&\gate{HZ(\tilde\theta(i))}&\meter{$r_i$}&&\arrow[dash,u,thick]&&&&  \\
&&&\phantom{random}\psi~\arrow[rightarrow,r,thick]&&&&&\rho^{out}_{\mathcal A}= P(\rho^{out}_{\mathcal B})&&&\arrow[dash,u,thick]&\qw\qwbundle{}&\gate{HZ(\tilde\theta(i))}&\meter{$r_i$}& 
\end{quantikz}
    }

    \caption[Resource $\sboqcp\simb$]{Pictorial representation of $\sboqcp\simb$
        (\Cref{pro:simulator}).  The variables correspond to those
        defined in \Cref{pro:simulator}. See \Cref{fig:convs1} for the notation;
        additionally,
    $\tilde\theta(i)=
    \begin{cases}
          \delta_i-\phi'_i, & \text{if}~i\in\vera \\
          \delta_i-\psi'_i, & \text{if}~i\in\vero. 
    \end{cases}$
    }
    \label{fig:simulator}
    \end{figure}
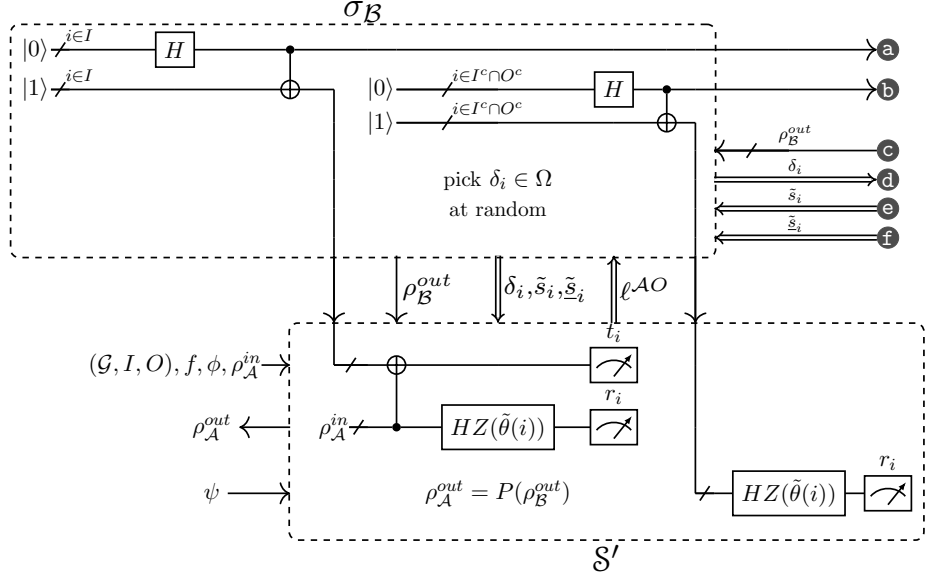

   Finally, since both figures have the same inputs and outputs, it remains to
   show that the leak is given by $\ell^{\mathcal{AO}}$. The leak of
   information $\ell^{\mathcal{AO}}$ is inputted to the simulator from Bob's
   interface. The simulator does not use the leak to create any useful
   information.  Now we need to prove that the simulator does not learn
   anything beyond $\ell^{\mathcal{AO}}$. For that, the information received by
   the simulator must be independent of the computation. First, $\delta_i$ is
   an arbitrary angle, thus independent of the computation. Second,
   $\circledl{a}$ and $\circledl{b}$ are completely mixed states because of the
   randomness $\alpha_i,\beta_i$; thus, the simulator cannot guess $\alpha_i$
   or $\beta_i$ with complete certainty.  A curious Bob might correctly guess
   other information such as the flow $f$, however it gives him no advantage. Therefore, there is no more leak than $\ell^{\mathcal{AO}}$ throughout
   the protocol. This concludes the proof.

   \end{proof}

   Note that our strategy for constructing the simulator $\simb$ in
   \Cref{pro:simulator} comes to us from the work on composable security of
   delegated quantum computation by Dunjko \etal in
   Ref.~\cite{dunjko2014composable}. However, its usefulness in our case is
   only guaranteed by \Cref{thm:correctness_boqc}.

   Here we consider the general case when Alice's input and output are quantum,
   \ie $\tilde I=I$ and $\tilde O=O$. The security is maintained as long as
  the input-output configuration admits the security model $\sboqcp$.  Thus, the
   security is maintained for $\tilde I\subset I$ and $\tilde O\subset O$
   for the following reasons: 

   First, letting $\tilde I\subset I$, we denote Alice's quantum input
   $\rho^{in}\in\mathcal H_{\tilde I}$ and quantum input of the security model
   $\rho^{in}_{\mathcal A}\in\mathcal H_{I}$. Since $\mathcal
   H_{I}\subset\mathcal H_{\tilde I}$, she has another classical input in the
   form of a bit string $c$ with length $\abs*{I\setminus\tilde I}$. As one may
   always encode classical information into qubits as we choose, we can set
   $\rho^{in}_{\mathcal A}=\rho^{in}\bigotimes_{i\in I\setminus
   I}\ketbra{c_i}$, where $c_i\in\{0,1\}$. Therefore, $\rho^{in}_{\mathcal A}$
   is now entirely quantum as modeled by $\sboqcp$. Finally, the same applies
   for $\tilde O\subset O$ in which classical output can be represented as a
   diagonal density matrix.


    Now it remains to prove the blindness of BOQCo. The proof is rather
    straightforward because BOQC and BOQCo differ only in the ordering among
    qubit transmissions, entanglements, and measurements.
    \Cref{thm:boqco_secure} provides the security statement of BOQCo, whose
    relaxation $\sboqcp\simb$ is provided in \Cref{pro:simulatoro}, in
    \Cref{app:protocols}.

\begin{theorem}\label{thm:boqco_secure}
    The BOQCo protocol with dishonest Bob $\pi_{boqco}'=(\conv{A},\conv{O})$,
    defined in \Cref{pro:boqco} is perfectly blind, and realizes ideal resource
    $\sboqcp$ defined in \Cref{fig:sboqcp}.

\end{theorem}
\begin{proof}
    By \Cref{def:security}, statement $\euscr R\xrightarrow{\pi_{boqco},0}\sboqcp$ is achieved
    if there exists a simulator $\simb$ such that $d(\conv{A}\conv{O}\euscr R,\sboqcp\simb)=0$.  Such a relaxation
    can be straightforwardly derived from \Cref{pro:simulator} by rearranging
    the process order, viz., qubit transmissions, entanglements, and
    measurements according to the lazy computation in \Cref{alg:lazy} (the
    relaxation is shown explicitly in \Cref{pro:simulatoro}, \Cref{app:protocols}). 

    Since BOQCo and BOQC differ only by the process order, they have the same
    configuration of inputs and outputs.  Therefore, $\conv{A}\conv{O}\euscr R$ and
    $\sboqcp\simb$ have pictorial representations as \Cref{fig:convs2} and
    \Cref{fig:simulator}, respectively, where the figures are proven to be
    indistinguishable in \Cref{thm:boqc_blindness}. The difference in the
    process order will not reveal any information about the computation because
    the ordering is determined before the protocol starts, arranged in
    pre-protocol steps (\Cref{def:prestep_boqc}).  

    Since we obtain an indistinguishable relaxation to $\conv{A}\conv{O}\euscr
    R$ that has the same leak as BOQC, this concludes the proof.
\end{proof}

\subsection{The BOQC-compatible graph states}

The BOQC and the BOQCo protocols provide simple cooperation between Alice and
Oscar to delegate their blind computations, allowing the malicious Bob to learn
no more that public information $\ell^{\mathcal{AO}}\coloneqq\{(\mathcal G,
\tilde I, \tilde O), \vera, \vero,\succ,>,b\}$.  The leaking information
$\ell^{\mathcal{AO}}$ is insufficient to infer the computation, but there is a
catch:

Consider an example of a real-life situation in which Oscar's company is
well-known for storing massive confidential databases. Thus, Bob might reasonably make an
priori assumption that they are running a quantum search algorithm.  If the
oracle's graph varies according to the request being made, Bob might infer some
information about Alice's request, such as, a simple graph marks a zeros state in
the Grover algorithm. Thus, Oscar's oracle graphs must remain the same for all
requests; he is only allowed to vary only the measurement angles for different queries. We call such graphs, \ie
a class of graphs that runs a set of different requests, ``BOQC-compatible''
graphs. 

A BOQC-compatible graph is a standard oracle graph, determined by Oscar, for a
class of requests. One can use a universal graph, \eg a fixed-size brickwork
state as introduced in Ref.~\cite{broadbent2009universal}; however, the number of
qubits increases rapidly with the circuit depth of the gate model representation.
Another strategy is to optimize the graphs, constrained to the set of requests
as done in Ref.~\cite{gustiani2019three}, for 2- and 3-qubits exact Grover
algorithms. The resulting graph is significantly more compact than the
brickwork graph.

\section*{Conclusion}
Here, we have introduced BOQC, a secure client-server oracle computation scheme with three
parties, that is an extension of the UBQC.  
Our scheme provides a means for a
client, without a quantum computer and without the capability of performing the oracle evaluation, to
securely cooperate with an oracle client to delegate the oracular quantum algorithm to a malicious quantum server.
We formally prove that BOQC is blind within composable definitions using the
Abstract Cryptography framework. 
A drawback of our security definition is that it
requires BOQC-compatible graphs as graph states.  
We extend BOQC to BOQCo
(BOQC-optimized), which adapts the protocol for efficient use on a solid-state quantum network, which has the attributes that the
server's qubits possess permanence and can be rapidly re-initialized. 
We provide explicit BOQC and BOQCo protocols that deal with both quantum and classical input and output. 
The BOQCo
scheme allows the server to employ a minimal number of physical qubits; in separate work we show how the scheme enables the implementation of Grover and Simon algorithms in a realistic NV-center network. 
While our protocols promise blindness, we cannot tell whether the
server is being malicious; thus, integrating verifiability into our protocol 
would be a desirable next step in this work.


\section*{Acknowledgement}
We gratefully acknowledge discussions on Abstract Cryptography with
C.~Portmann; we also thank him for proofreading the security proofs in this
manuscript. C.~G. thanks E.~Sovetkin for proofreading some lemmas and for
his idea to employ measurement calculus in reasoning correctness. C.~G.
acknowledges support from a research collaboration with Huawei Technologies.

\printbibliography

\begin{appendix}

\section{Proof of lemmas and theorems}
\label{app:proofs}

All proofs in this section are similar to the ones in Ref.~\cite{gustiani}.

The following proposition is required to proof \Cref{lem:correction}.

\begin{proposition}\label{pro:acausal}\cite{gustiani}
    Suppose $i$ and $j$ are two distinct vertices in an open graph $(\mathcal G,I,O)$
    with flow $(f,\succ)$, then $f(i)\neq f(j)$.
\end{proposition}
\begin{proof}
    Let us assume the contrary, there exists three distinct nodes
    $i,j,k\in \mathcal G$ such that $f(i)=f(j)=k$. 
    Given that flow $f$ induces a partial ordering $\succ$; by the definition of
    flow, every criterion in it must be satisfied, \ie 
    \begin{align*}
        \text{\bfseries (F0)}&~(k,i)\in E(\mathcal G)~\text{and}~(k,j)\in E(\mathcal G) \\
        \text{\bfseries (F1)}&~k \succ i~\text{and}~k\succ j\\
        \text{\bfseries (F2)}&~\forall v\in \neig{k}\setminus\{i\}, v\succ i~\text{and}~
                                \forall w\in \neig{k}\setminus\{j\}, w\succ k.
    \end{align*}
    The first criterion of \textbf{(F2)} entails $j\in \neig{k}\setminus \{i\}$
    since $j\in \neig{k}$ by criterion \textbf{(F0)}, which implies $j\succ i$;
    hence, \textbf{(F1)} imposes an ordering $k\succ j\succ i$.
    However, the second criterion of \textbf{(F2)} also implies that $i\in \neig{k}\setminus \{j\}$,
    because $i\in \neig{k}$ --- \textbf{(F0)}; thus, $i\succ j$, which 
    leads to a contradiction.
\end{proof}

\lemcorrection*
\begin{proof}[Proof of \Cref{lem:correction}]
Observe that both patterns are ordered in the following manner: (1) graph state
preparation (qubit initialization and entangling operations), (2) adaptive
measurements that follow the partial order $\succ$, and (3) applied Pauli
corrections interspersed among the measurements. It is evident from inspection
that the measurement operators $M_i^{\phi_i}$ occur in the same order in
patterns $\mathfrak P_1$ and $\mathfrak P_2$. Thus, it only remains to prove that
both patterns implement identical Pauli corrections. 

Consider the $X$-corrections. After measuring $i$ with output $s_i$, the
$X$-correction propagates to qubit $f(i)$, that is, $X_{f(i)}^{s_i}$ will appear
before measuring $f(i)$, as seen in $\mathfrak P_1$. Equivalently,\footnote{Note that
$f$ is injective, which is shown in \Cref{pro:acausal}.}
qubit $f(i)$
receives $X$-correction from measurement outcome of qubit $i$, where $i$ is the
preimage of $f(i)$. Thus, before measuring $f(i)$, the correction from
measuring $i$ must be present, namely $X_{f(i)}^{s_i}$. Let $f(i)\eqqcolon j$,
thus $X_{f(i)}^{s_i}=X_j^{s_{\invf{j}}}$, which is the corresponding operator
in $\mathfrak P_2$.

The same treatment applies to the $Z$-corrections; however, the corrections
propagate to a subset of nodes instead of a single node. From pattern $\mathfrak P_1$,
measuring $i$ produces $Z$-corrections on nodes
$\{\neig{f(i)}\setminus\{i\}\}\eqqcolon K$, which are the neighbors of
$f(i)$. From the perspective of $k\in K$, qubit $k$ receives corrections from all
measurement outcomes of $i$, where the set of neighbors of $f(i)$ contains
$k$.  From the flow condition, we know that $i\prec k$, meaning that $i$ has already
been measured. Thus, the $Z$-correction of qubit $k$ must include all
measurement outcomes from qubits $\{i\prec k \mid k\in\neig{f(i)}\}$. That is
what we see in the superscript of $Z$ in pattern $\mathfrak P_2$.   

Since every node is corrected by one Pauli operation, the correction of the
output can be represented as an operator acting on the output qubits.
Note that the output qubits  lie in the last layer of the graph: for all $k\not\in
O$, for all $i\in O$, we have $k\prec i$; thus, the corrections can appear at
the very end of the pattern.
\end{proof}

\thmcorrectnessboqc*
\begin{proof}[Proof of \Cref{thm:correctness_boqc}]

  Given are Alice's input $\{(\mathcal G,I,O),\phi\}$ together with a flow $f$,
  and Oscar's input $\psi$. We can write the total computation as $\{(\mathcal
  G,I,O),\theta\}$, where $\theta=\phi\cup\psi$ contains measurement angles for
  all nodes $i\in O^c$.  We need to show that  \Cref{pro:boqc}  results in
  pattern 
   $\POP_{i\in O^c}
   (X^{s_i}_{f(i)}\prod_{k\in N_{\mathcal G}(f(i)) \setminus \{i\}}
   Z_k^{s_i}M_i^{\theta_i})E_{\mathcal G} N_{I^c}^0\eqqcolon\mathcal P_{\euscr S}$.  

  First, we consider Alice and Oscar as one client called super-Alice.
  Omitting the randomness $r_i=\alpha_i=\beta_i=0$,\footnote{In this condition
  the computation is within the 1WQC scheme.} running a computation
  with     \Cref{pro:boqc} results in the pattern 
  \[\mathcal P_{\euscr R}=
  \bigotimes_{j\in O} 
    X_j^{s_{\invf{j}}}
    Z_j^{z(j)}\,
  \POP_{i\in O^c}
  M_i^{\theta'_i}
    X_i^{s_{\invf{i}}}
    Z_i^{z(i)}
  E_{\mathcal G}N_{I^c}^0
  \stackrel{\text{Lem.~\ref{lem:correction}}}{=}
  \mathcal P_{\euscr S},\]
  where $z(i)\coloneqq {\sum_{k\prec i,i\in\neig{f(k)}}s_k}$.
  Thus, with the absence of randomness, $\mathcal
  P_{\euscr R}=\mathcal P_{\euscr S}$. 

  Second, consider $\gamma$, a random variable that contains $\alpha$ and
  $\beta$, thus $\gamma=\alpha\cup\beta$. Including $\gamma$ and $r$ in the
  computation,\footnote{ This can be seen as UBQC with graph $\mathcal G$
  instead of the brickwork graph~\cite{broadbent2009universal} 
.} we obtain pattern 
\begin{equation}\label{eq:super-alice}
     \mathcal P_{\euscr R}\coloneqq
     \POP_{i\in O^c} (X^{s_i}_{f(i)}\prod_{k\in N_{\mathcal G}(f(i)) \setminus
  \{i\}} Z_k^{s_i+r_i}M_i^{\theta_i+\gamma_i+\pi r_i})E_{\mathcal G}
  N_{I^c}^\gamma.
\end{equation}

Thus, using \Cref{eq:correction} and commutation $E_{\mathcal G}N_i^{\gamma_i}
=E_{\mathcal G}Z_i(\gamma_i)N_i^0 =Z_i(\gamma_i)E_{\mathcal G} N_i^0,$ random
variable $\gamma_i$ is cancelled out.\footnote{ Within the gate representation, measuring $i$ in $\gamma$ means
$M_i^{\gamma}\rho=\tr[P_i H_i Z_i(-\gamma)\rho]$, where $P$ is measurement
projector in the computational basis; thus, $M_i^\gamma N_i^\gamma = M_i^0N_i^0$.}
The random angle $\pi r_i$ is canceled out by
flipping the measurement outcome, $s_i=s_i\oplus r_i$, which is expressed in
the superscript of $Z$-corrections. 
Thus, equality $\mathcal P_{\euscr R}=\mathcal P_{\euscr S}$ is maintained.

Third, we show that encryption of the quantum inputs
$X_i^{t_i}Z_i(\alpha_i)(\rho^{in})$ is perfectly decrypted during the process.
Random variable $\alpha_i$ is cancelled out in the same manner as removing
$\gamma$ above.  The random variable $t_i$ is canceled out by flipping the sign
of the measurement angle \footnote{This can be considered as a measurement that is $X$-dependent.  }
$M_i^\theta
X_i\rho^{in}=M_i^{-\theta}\rho^{in}$. Thus, we conclude that $\rho^{in}$ is
intact.

Finally, we divide super-Alice into Alice and Oscar. Calculating corrected angles
is separately done, since they know the real measurement outcomes
$s_i=\tilde s_i\oplus r_i$. Canceling out the $z$-rotations $R^z_i(\alpha_i)$ or
$R^z_i(\beta_i)$ is also done separately, see
lines~\ref{ln:a_delta} and~\ref{ln:o_delta} of the Protocol. Thus, Alice and Oscar keep their
measurement angles and the random variables to themselves. 
\end{proof}

\lemmutuallydisjoint*
\begin{proof}[Proof of \Cref{lem:mutually_disjoint}]
    We need to prove that $f(i)\in A(i)$ and $A(i)\cap A(k)=\varnothing, k\neq
    i$, for all $i,k\in O^c$.

    Proof of the first part: Using the existence of flow $f:O^c\rightarrow
    I^c$, thus $f(i)\eqqcolon j$ exists (where $j\in I^c$), and \Cref{eqn:flow}
    applies. It follows that $j\in \neig{i}$, and also $j\in\neigc{i}$, where
    $j\succ i$. But $j\not\in \cup_{k< i}\neigc{k}$, because the flow condition
    $k\in\neig{j}\setminus\{i\}, k\succ i$ (it is also true that
    $k>i$);\footnote{The flow must not have a neighbour in the past, \ie a
    neighbour that has already measured.} then clearly $j\not\in I$. Thus,
    $A(i)$ contains at least $f(i)$.

    Proof of the second part: 
    Denote $a_i\coloneqq \neigc{i}$,
    $b_i\coloneqq I\cup_{j<i}\neigc{j}$, thus $A(i)\equiv a_i\setminus b_i$.
    Given that every node has an ordering, consider the case $i<k$:
  \[
      A(i)\cap A(k) \equiv (a_i\setminus b_i)\cap (a_k\setminus b_k)
    = [( a_i\setminus b_i)\setminus b_k]\cap a_k
    \stackrel{b_i\subset b_k}{=}  (a_i\setminus b_k)\cap a_k
    \stackrel{a_i\subset b_k}{=}\varnothing. 
  \]
It is clear that $b_i\subset b_k$ and $a_i\subset b_k$, since $i<k$. 
\end{proof}

\lempartition*
\begin{proof}[Proof of \Cref{lem:partition}] 
    Let $\mathcal G=(V,E)$.
    Recall that $\cup_{i\in V}A(i)=I^c\iff \cup_{i\in V}A(i)\subseteq I^c$ and
    $\cup_{i\in V}A(i)\supseteq I^c$; we assume the contrary, that is, let
    $\cup_{i\in V}A(i)=S$, thus $S\neq I^c$, \ie either $S\subsetneq I^c$ or
    $S\supsetneq I^c$ is true.

    Consider the first case, where $S\subsetneq I^c$.  Here, there exists $k\in
    I^c$, such that $k\not\in A(i)$ for all $i\in I^c$. {We split this first case proof into two
    parts: for~$i\in O^c$ and for~$i\in O$.}

    {For $i\in O^c$}, by using \Cref{lem:mutually_disjoint}, $A(i)$ contains at
    least $f(i)$. Note that $f:O^c\rightarrow I^c$; thus, $\invf{k}\eqqcolon j$
    exists, where $j\in O^c$. Therefore, $k\in A(j)$. Moreover, from
    \Cref{lem:mutually_disjoint} we know that $S\neq\varnothing$. 

    {For $i\in O$, we have $A(i)\subset O$ because if $A(i)$ contains
        an element in $O^c$, it is covered in the case $i\in O^c$ above. As
        output nodes also have a total ordering, $\cup_{i\in O} A(i)$ covers all
        output nodes that are disjoint to $O^c$ nodes.}

    We arrive at a contradiction, so the assumption is incorrect, and we
    conclude that $I_{c} \setminus S = \varnothing$ or $I_{c}\subseteq S$.

    Consider the second case, where $S\supsetneq I^c$. Here, there exists a $k\in
    I^c$, such that $A(k)$ contains $m$ where $m\not \in I^c$. By definition,
    $A(k)$ contains at most its closed neighborhood excluding the inputs, \viz
    $\neig{k}\setminus I$. Thus, $A(k)$ must be in the graph, since $\mathcal
    G[\neig{k}]$ is a subgraph of $\mathcal G$; also, any element of $I$ cannot
    be in $A(k)$. We again arrive at a contradiction, and must conclude that
    $S \setminus I_{c} = \varnothing$ or $S \subseteq I_{c}$.

    Since $S \subseteq I^{c}$ and $I^{c} \subseteq S$, it follows that $S=I^c$.
\end{proof}

\leme*
\begin{proof}[Proof of \Cref{lem:e}]
    Note that we are employing the notation of \Cref{eq:e_notation}.  Denote
    $(i,j)$ as an edge in $E(\mathcal G)$ with end nodes $i$ and $j$. If we represent
    $\prod_{i\in V} E^>_{i\neig{i}}$ as a collection of edges, we obtain
    $\bigcup_{i\in V} \{(i,k)\mid k\in\neig{i}, k>i \}\eqqcolon S$.  On the
    other hand, the Handshaking lemma \cite{biggs1986graph} implies that,
    $\biguplus_{i\in V} \{(i,k)\mid k\in\neig{i}\}$ will result in a multiset
    that contains double copies of edges --- $\forall a\in E(\mathcal G)$ the multiplicity of $a$ is 2,
    where $\uplus$ signifies the union of multisets
    \cite{syropoulos2000mathematics}. Thus restricting to edges $(i,k)$ where
    $i<k$ eliminates the double counting; therefore $S=E(\mathcal G)$.
\end{proof}

\thmcorrectnessboqco*
\begin{proof}[Proof of \Cref{thm:correctness_boqco}]
    We use the same reasoning as in the proof of \Cref{thm:correctness_boqc}
    by showing that the resulting pattern from
    \Cref{pro:boqco} can be reduced to the pattern in
    \Cref{thm:pattern}.  Note that, in this context, converters
    $\conv{A},\conv{O}$, and $\conv{B}$ correspond to
    procedures of the BOQCo protocol defined in \Cref{pro:boqco}; real resource $\euscr R$ comprises the same
    elements as the BOQC \viz a secure key, two-way
    classical channels, and two-way quantum channels.

    Consider Alice and Oscar as one party, which we call super-Alice, who has all
    information about the angles $\phi,\psi$ and the random variables $\alpha,
    \beta, r, t$. We denote $\gamma\coloneqq\alpha\cup\beta$ and $\theta\coloneqq\phi\cup\psi$. 

    First, we omit randomness, so that $r_i=t_i=\alpha_i=\beta_i=0$ for
    all $i$; the resulting BOQCo pattern ($\mathfrak P_{boqco}$) can be written as 
    \[\mathfrak P_{boqco}\coloneqq 
    \bigotimes_{j\in O}
    X_j^{s_{\invf{j}}}
    Z_j^{z(j)}
    E_O
    \OP_{i\in O^c}  
    M_i^{\theta_i}
    X_i^{s_{\invf{i}}}
    Z_i^{z(i)}
    E_{i\neig{i}}^>
    N_{A(i)}^0
    =\mathfrak P_{lazy}
    \stackrel{\text{Theo.} \ref{thm:lazy_correctness}}{=}
    \mathfrak P_{1wqc},
\]
    where $z(i)\coloneqq \bigoplus_{k<i, i\in\neig{f(k)}}s_k$,
    $\invf{i}\equiv{f^{-1}(i)}$,
    $\mathfrak P_{lazy}$ is the resulting pattern of lazy 1WQC (\Cref{eq:pattern_lazy}),
    and $\mathfrak P_{1wqc}$ is the resulting pattern of 1WQC (\Cref{eq:pattern}), where 
\[\mathfrak P_{1wqc}\coloneqq \POP_{i\in O^c}(X_{f(i)}^{s_i}
\prod_{k\in\neig{f(i)}\setminus\{i\}}Z_i^{s_i} M_i^{\theta_i})E_{\mathcal G}N_{I^c}^0.\]
    Thus, $\mathfrak P_{boqco}=\mathfrak P_{1wqc}$ holds in the absence of randomness.
    Note that total ordering $>$ is consistent with partial ordering $\succ$ (partial ordering
    induced by the flow); thus we can interchangeably use both.
    Recall that operator $\bigotimes$ allows for concurrency, whereas $\prod$, $\POP$, and $\OP$ indicate serial operations.

    Second, we introduce random variables $\gamma$ and $r$ into the protocol; now, pattern
    $\mathfrak P_{boqco}$ becomes
    \begin{equation}
        \mathfrak P_{boqco}= 
        \POP_{i\in O^c}
        (X^{s_i}_{f(i)}\prod_{k\in N_{\mathcal G}(f(i)) \setminus
        \{i\}} Z_k^{s_i+r_i}M_i^{\theta_i+\gamma_i+\pi r_i})E_{\mathcal G}
        N_{I^c}^\gamma,
    \end{equation}
    which is identical to \Cref{eq:super-alice}. Thus, from this point, the proof proceeds identically  
    to that of \Cref{thm:correctness_boqc} from \Cref{eq:super-alice}. 
\end{proof}

\section{Protocols and relaxations}\label{app:protocols}
\begin{protocol}[h]\caption{BOQC with classical input-output}
    \label{pro:boqcc}
  \begin{algorithmic}[1]
      \Statex{\hspace{-2em}\bfseries Alice's input: $\{(\mathcal G, I, O),f,\phi,c=c_1c_2,\dots,c_n\}$}\Comment{$\tilde I=\tilde O=\varnothing,\rho^{in}_{\mathcal A}=\prod_{i=1}^{n}\ketbra{c_i}$}
  \Statex{\hspace{-2em}\bfseries Oscar's input: $\{\psi\}$}
  \Statex{\hspace{-2em}\bfseries Alice's output for an honest Bob:
  $\rho^{out}_{\mathcal A}=\mathcal E(\rho^{in}_{\mathcal A})$}\Comment{$\rho^{out}_{\mathcal A}$ is a diagonal matrix}

  \Statex{\itshape Assumptions and conventions: 
          \begin{enumerate}[(I)]
              \item Alice ($\mathcal A$) and Oscar ($\mathcal O$) have
                  performed pre-protocol steps in \Cref{def:prestep_boqc}; 
                  Bob knows $\{(\mathcal G,\tilde I,\tilde O),\vera,\vero,\succ,>,b\}$. 
                  Here, we set $\tilde I=\tilde O=\varnothing$, and recall 
                  $\Omega=\{\frac{\pi k}{2^{b-1}}\}_{0\leq k<2^b}$.

      \item $s_{\invf{i}}=0,\forall i\in I$.
          
      \item $\invf{i}\equiv f^{-1}(i)$ and 
          $z(i)\coloneqq\bigoplus_{k\prec i, i\in \neig{f(k)}}s_k$. 
          \end{enumerate}}


  \Statex{\hspace{-2em}\bfseries \circled{0} Pre-preparation}

  \State{Alice and Oscar receive a key $r$ via a secure key channel,
  where $r_i\in\{0,1\}$, for $i\in O^c$.}

    \Statex{\hspace{-2em}\bfseries\centering \circled{1} State preparation}

    \For{$i \in V$ which follows partial ordering $\succ$}
    \If{$i\in \vera $} 
    \If{$i \in I$}
    \State{Alice prepares $\ket{+_{\alpha_i+c_i\pi}}_i$  and sends it to Bob; $\alpha_i\in\Omega$ is chosen at random.}
  \Else
      \State\parbox[t]{0.89\linewidth}{Alice prepares
        $\ket{+_{\alpha_i}}_i$ and sends it to Bob, where $\alpha_i\in\Omega$ is
      chosen at random}
      \EndIf
    \ElsIf { $i \in \vero$}
      \State{Oscar prepares $\ket{+_{\beta_i}}_i$ and sends it to Bob; $\beta_i\in\Omega$ is chosen at random.}
    \EndIf
    \EndFor
    \Statex{\hspace{-2em}\bfseries\centering \circled{2} Graph state formation}
    \State{Bob applies entangling operator $E_{\mathcal G}$ defined in \Cref{eq:e_notation}.}

    \Statex{\hspace{-2em}\bfseries \circled{3} Classical interaction and measurement}

    \For{$i \in V$ which follows partial ordering $\succ$}\Comment{\eg it follows $>$}
    \If{$i\in \vera$} 

    \State{Alice computes $\phi'_i=(-1)^{s_{\invf{i}}}\phi_i+z(i)\pi$.}
    \State{Alice computes $\delta_i\coloneqq \phi'_i + \pi r_i + \alpha_i$, and sends Bob $\delta_i$.} 

    \ElsIf { $i \in \vero$}
     \State{Oscar computes $\psi'_i=(-1)^{s_{\invf{i}}}\psi_i+ z(i)\pi$.}
    \State{Oscar computes $\delta_i\coloneqq \psi'_i + \pi r_i + \beta_i$, and sends Bob $\delta_i$.}

    \EndIf

    \State{Bob measures qubit $i$ in basis $\ket{\pm_{\delta_i}}$, then sends Alice and Oscar the outcome $\tilde s_i$.}

    \State{Alice and Oscar set $s_i=\tilde s_i\oplus r_i$.}
    \EndFor
 \end{algorithmic}
\end{protocol}
\begin{protocol}[h]
    \caption{BOQCo with classical input-output}
  \label{pro:boqcocc}
  \begin{algorithmic}[1]

      \Statex{\hspace{-2em}\bfseries Alice's input: $\{(\mathcal G, I, O),f,\phi,c=c_1c_2,\dots,c_n\}$}\Comment{$\tilde I=\tilde O=\varnothing,\rho^{in}_{\mathcal A}=\prod_{i=1}^{n}\ketbra{c_i}$}
  \Statex{\hspace{-2em}\bfseries Oscar's input: $\{\psi\}$}
  \Statex{\hspace{-2em}\bfseries Alice's output for an honest Bob:
  $\rho^{out}_{\mathcal A}=\mathcal E(\rho^{in}_{\mathcal A})$}\Comment{$\rho^{out}_{\mathcal A}$ is a diagonal matrix}

  \Statex{\itshape Assumptions and conventions: 
          \begin{enumerate}[(I)]
              \item Alice ($\mathcal A$) and Oscar ($\mathcal O$) have
                  performed pre-protocol steps in \Cref{def:prestep_boqc}; 
                  Bob knows $\{(\mathcal G,\tilde I,\tilde O),\vera,\vero,\succ,>,b\}$. 
                  Here, we set $\tilde I=\tilde O=\varnothing$, and recall 
                  $\Omega=\{\frac{\pi k}{2^{b-1}}\}_{0\leq k<2^b}$.

      \item $s_{\invf{i}}=0,\forall i\in I$.
          
      \item $\invf{i}\equiv f^{-1}(i)$ and 
          $z(i)\coloneqq\bigoplus_{k\prec i, i\in \neig{f(k)}}s_k$. 
          \end{enumerate}}

\Statex{\hspace{-2em}\bfseries \circled{0} Pre-preparation}
  \State{Alice and Oscar receive a key $r$ via a secure key channel,
  where $r_i\in\{0,1\}$, for $i\in O^c$.}

\Statex{\hspace{-2em}\bfseries\centering \circled{1} BOQC by parts}
%
\For{$i \in V\setminus O$ with ordering $>$}

\For{$k\in A(i)\cup I $}\Comment{see \Cref{eq:ai} in \Cref{sec:1wqc}}\label{ln:boqcocc_mainloop}
\If{$k\in I$}\Comment{input qubit}

\State{Alice prepares $\ket{+_{\alpha_i+c_i\pi}}_i$ and sends it to Bob;
    $\alpha_i\in\Omega$ is chosen at random.
}

\Else \Comment{auxiliary qubit}
            \If{$k\in \vera$}
                    \State{Alice prepares $\ket{+_{\alpha_k}}_k$ and sends
                    it to Bob, $\alpha_k\in\Omega$ is chosen at random.}
            \ElsIf{$k\in\vero$}
                    \State{Oscar prepares $\ket{+_{\beta_k}}_k$ and send it to Bob, $\beta_k\in\Omega$ is chosen at random.}
            \EndIf

        \EndIf
    \EndFor

    \State{Bob applies entangling operation ${E^>_{i N_{\mathcal G}(i)}} $.}
    \Comment{See \Cref{eq:e_notation} in \Cref{sec:1wqc}} 

    \If{$i\in \vera$}
        \State{Alice computes $\phi'_i=(-1)^{s_{\invf{i}}}\phi_i + z(i)\pi$.}
        \State{Alice computes $\delta_i\coloneqq \phi'_i + \pi r_i + \alpha_i$ and sends Bob $\delta_i$.}\label{ln:boqcocc_alice_delta}

    \ElsIf{$i\in\vero$}
        \State{Oscar computes $\psi'_i=(-1)^{s_{\invf{i}}}\psi_i + z(i)\pi$.}
        \State{Oscar computes $\delta_i\coloneqq \psi'_i + \pi r_i + \beta_i$ and sends Bob $\delta_i$.} 
    \EndIf

    \State{Bob measures $i$ in $\ket{\pm_{\delta_i}}$ basis, then sends Alice and Oscar 
          the outcome $\tilde s_i$.}
    \State{Alice and Oscar set $s_i=\tilde s_i\oplus r_i$.}
\EndFor
  \end{algorithmic}
\end{protocol}

\begin{protocol}[!h]
    \caption{\cite{gustiani}A relaxation $\sboqcp\simb$}
\label{pro:simulator}

\begin{algorithmic}[1]
  \Statex{\itshape Conventions: 

    \begin{enumerate}[(I)]
        \item Given public information $\{(\mathcal G,\tilde I,\tilde O),\vera,\vero,\succ,>,b\}$,
            where $\tilde I=I$, $\tilde O=O$, $\mathcal G=(V,E)$, 
            $V=\vera\cup\vero$, and $\Omega=\{\frac{\pi k}{2^{b-1}}\}_{0\leq k<2^b}$.
      \item $s_{\invf{i}}=0, \underline s_{\invf{i}}=0$ for all $i\in I$ and $t_i=0$ for all $i\in I^c$.
      \item $\invf{i}\equiv f^{-1}(i)$, $z(i)\coloneqq\bigoplus_{k<i, i\in \neig{f(k)}}s_k$,
            $\underline z(i)\coloneqq\bigoplus_{k<i, i\in \neig{f(k)}}\underline s_k$,
          and  $t(i)\coloneqq\sum_{k\in I, i\in N_{\mathcal G}(k)}t_k$.
      \item Measurements are performed in the computational basis.
    \end{enumerate}}

\Statex{\bfseries The simulator $\simb$}
\State {Prepares an EPR pair $(\ket{00}+\ket{11})/\sqrt{2}$ for every node  $i\in O^c$ and outputs its half.}

\State{Picks random angles $\{\delta_i\in\Omega \mid i\in O^c\}\eqqcolon\delta$
    and outputs it.}
\State{Receives responses $\{\tilde s_i\in\{0,1\} \mid
i\in O^c\cap\vera\}\eqqcolon \tilde s$ and
$\{\underline{\tilde s}_i\in\{0,1\} \mid i\in O^c\cap\vero\}\eqqcolon \underline{\tilde s}$.}

\State{Receives the corresponding output qubits $\rho^{out}_{\mathcal B}$ (all qubits $i\in O$).}

\State{Sends $\sboqcp$ the other halves of EPR pairs 
    ($\{$half-EPR-$i\mid i\in O^c\}\eqqcolon$EPRs), 
$\delta, \tilde s, \tilde{\underline s}$, and $\rho^{out}_{\mathcal B}$.}\label{ln:outinfo}

\Statex{}
\Statex{\bfseries The ideal BOQC resource $\sboqcp$}
\State{Receives $\{(\mathcal G,I,O),f,\rho^{in}_{\mathcal A},\phi\}$ at Alice's interface and $\psi$ at
Oscar's interface, and information from step \ref{ln:outinfo} ---  
EPRs, $\delta, \tilde s, \tilde{\underline s}$, and $\rho^{out}_{\mathcal B}$.}

\State{Applies CNOT gates between $\tr_{I\setminus i}[\rho^{in}_{\mathcal A}]$ (as control) and half-EPR-$i$, for all $i\in I$, 
    then measures the half-EPR, stores the outcome as $t_i$, and 
     updates measurement angles: 
    \[\begin{gathered}
    \forall i\in I, \phi_i=(-1)^{t_i}\phi_i, \\ 
    \forall i\in I,\forall j\in \neig{i}\cap\vera,\phi_j=\phi_j+t_i\pi\\ 
    \forall i\in I,\forall j\in \neig{i}\cap\vero,\psi_j=\psi_j+t_i\pi. 
    \end{gathered}\]
}\label{ln:sim_update_angles_ti}

\For{$i\in O^c$ which follows ordering $>$}
\If{$i\in\vera$}

\State{Computes $\phi'_i=(-1)^{s_{\invf{i}}}\phi_i + z(i)\pi$.} 
\State{Computes $\theta'_i=\delta_i-\phi'_i$.}

\ElsIf{$i\in\vero$}
\State{Computes $\psi'_i=(-1)^{\underline s_{\invf{i}}}\psi_i + \underline z(i)\pi$.}
\State{Computes $\theta'_i=\delta_i-\psi'_i$.} 
\EndIf


\If{$i\in I$}
\State{Applies $HZ_i(\theta_i)$ to an input qubit
$\tr_{I\setminus i}[\rho^{in}_{\mathcal A}]$ followed by measurement.}
\Else
\State{Applies $HZ_i(\theta_i)$ to the half-EPR-$i$ followed by measurement}
\EndIf

\State{Stores the measurement outcome as $r_i$ and 
sets $s_i=\tilde s_i \oplus r_i$ and $\underline s_i=\tilde{\underline s}_i\oplus r_i$.}
\EndFor

\State{Corrects the output $P(\rho^{out}_{\mathcal B})\eqqcolon 
\rho^{out}_{\mathcal A}$, where 
$P\equiv \bigotimes_{i\in O}X_i^{s_{\invf{i}}+t_i}
Z_{i}^{z(i)+t(i)}$,
then outputs it on Alice's interface.
}
  \end{algorithmic}
\end{protocol}
\begin{protocol}[!b]
    \caption{A relaxation ($\sboqcp\simb$) for BOQCo}
\label{pro:simulatoro}
  \begin{algorithmic}[1]

  \Statex{\itshape Conventions: 

    \begin{enumerate}[(I)]
        \item Given public information $\{(\mathcal G,\tilde I,\tilde O),\vera,\vero,\succ,>,b\}$,
            where $\tilde I=I$, $\tilde O=O$, $\mathcal G=(V,E)$, 
            $V=\vera\cup\vero$, and $\Omega=\{\frac{\pi k}{2^{b-1}}\}_{0\leq k<2^b}$.
      \item $s_{\invf{i}}=0, \underline s_{\invf{i}}=0$ for all $i\in I$ and $t_i=0$ for all $i\in I^c$.
      \item $\invf{i}\equiv f^{-1}(i)$, $z(i)\coloneqq\bigoplus_{k<i, i\in \neig{f(k)}}s_k$,
            $\underline z(i)\coloneqq\bigoplus_{k<i, i\in \neig{f(k)}}\underline s_k$,
          and  $t(i)\coloneqq\sum_{k\in I, i\in N_{\mathcal G}(k)}t_k$.
      \item Measurements are performed in the computational basis.
    \end{enumerate}}

\Statex{\bfseries Part \circled{1}}
\State{$\sboqcp$ receives $\{(\mathcal G,I,O),f,\rho^{in}_{\mathcal A},\phi\}$ at Alice's interface and $\psi$ at Oscar's interface. }
\Statex{}
\Statex{\bfseries Part \circled{2}}
\For{$i \in V\setminus O$ with ordering $>$}
\For{$k\in A(i)\cup I$}\Comment{see \Cref{eq:ai} in \Cref{sec:lazy}}


%

\Statex{\bfseries The simulator $\simb$}
\If{$k\not\in O$}

\State{Prepares an EPR pair $(\ket{00}+\ket{11})/\sqrt{2}$ and outputs its half.}
\State{Picks $\delta_k\in\Omega$ at random, outputs $\delta_k$.}

\State{Receives a responses $\tilde s_k$ and $\tilde {\underline s}_k$, where $\tilde s_k,\tilde{\underline s}_k\in\{0,1\}$.}
\State{Sends resource $\sboqcp$ the other half of EPR pair (half-EPR-$k$), $\delta_k$, $\tilde s_k$,and $\tilde{\underline s}_k$.}

\ElsIf{$k \in O$}
\State{Receives output qubit $\tr_{O\setminus k}[\rho^{out}_{\mathcal B}]$, and   
sends it to the ideal resource.}
\EndIf

\Statex{}
\Statex{\bfseries The ideal BOQCo resource $\sboqcp$}
\State{Receives \{half-EPR-$k$,$s_k,\tilde s_k,\tilde{\underline s}_k,\delta_k$ \} if $k\in O^c$, otherwise receives 
$\tr_{O\setminus k}[\rho^{out}_{\mathcal B}]$.}

\If{$k\in I$}
\State{Applies CNOT gate between $\tr_{I\setminus k}[\rho^{in}_{\mathcal A}]$ and half-EPR-$k$.}
\State{Measures half-EPR-$k$, stores the outcome as $r_k$, then 
       updates: 
    \[\begin{gathered}
        \phi_k=(-1)^{t_k}\phi_k, \\
    \forall j\in \neig{k}\cap\vera,\phi_j=\phi_j+t_k\pi\\ 
    \forall j\in \neig{k}\cap\vero,\psi_j=\psi_j+t_k\pi. 
    \end{gathered}\]}
\EndIf

\If {$k\in O^c$}
    \If{$k\in\vera$}
\State{Computes $\phi'_k=(-1)^{s_{\invf{k}}}\phi_k + z(k)\pi$.}
\State{Computes $\theta'_k=\delta_k-\phi'_k$.}
    \ElsIf{$k\in\vero$}
\State{Computes $\psi'_k=(-1)^{\underline s_{\invf{k}}}\psi_k + \underline z(k)\pi$.}
\State{Computes $\theta'_k=\delta_k-\psi'_k$} 
    \EndIf

    \State{Applies $HZ_k(\theta'_k)$ to half-EPR-$k$ followed by measurement.}
\State{Stores the measurement outcomes as $r_k$.}
\State{Sets $s_k=\tilde s_k\oplus r_k$ and $\tilde{\underline s}_k=\tilde{\underline s}_k\oplus r_k$ if $k\in\vero$.}

\Else \Comment{$k\in O$}
\State{Corrects $\tr_{O\setminus k}[\rho^{out}_{\mathcal B}]$ with $X_k^{s_{\invf{k}}+t_k} Z_{k}^{z(k)+t(k)}$ and outputs it
    on Alice's interface.}
\EndIf

\EndFor
\EndFor

\end{algorithmic}
\end{protocol}

\end{appendix}

\end{document}